\documentclass[conference]{IEEEtran}
\IEEEoverridecommandlockouts
\usepackage{cite}
\usepackage{amsmath, amsthm, amssymb,amsfonts}
\usepackage{multicol, hhline}
\usepackage{algorithmic}
\usepackage{graphicx}
\usepackage{textcomp}
\usepackage{xcolor}
\usepackage{physics}
\usepackage{tabularx}
\usepackage{adjustbox}

\def\BibTeX{{\rm B\kern-.05em{\sc i\kern-.025em b}\kern-.08em
    T\kern-.1667em\lower.7ex\hbox{E}\kern-.125emX}}

\newtheorem {theorem} {Theorem}

\newtheorem {lemma} {Lemma}

\graphicspath{ {./images/} }

\begin{document}

\title{Quantum Walk Random Number Generation: Memory-based Models}
%{\footnotesize \textsuperscript{*}Note: Sub-titles are not captured in Xplore and
%should not be used}
%\thanks{Identify applicable funding agency here. If none, delete this.}
%}

%\author{\IEEEauthorblockN{1\textsuperscript{st} Minwoo Bae}
\author{\IEEEauthorblockN{Minu J. Bae}
\IEEEauthorblockA{\textit{Computer Science Department} \\
\textit{University of Connecticut}\\
Storrs CT, USA \\
minwoo.bae@uconn.edu}
%\and
%%\IEEEauthorblockN{2\textsuperscript{nd} Given Name Surname}
%\IEEEauthorblockN{Walter O. Krawec}
%\IEEEauthorblockA{\textit{Computer Science Department} \\
%\textit{University of Connecticut}\\
%Storrs CT, USA \\
%walter.krawec@uconn.edu}
%\and
%\IEEEauthorblockN{3\textsuperscript{rd} Given Name Surname}
%\IEEEauthorblockA{\textit{dept. name of organization (of Aff.)} \\
%\textit{name of organization (of Aff.)}\\
%City, Country \\
%email address or ORCID}
%\and
%\IEEEauthorblockN{4\textsuperscript{th} Given Name Surname}
%\IEEEauthorblockA{\textit{dept. name of organization (of Aff.)} \\
%\textit{name of organization (of Aff.)}\\
%City, Country \\
%email address or ORCID}
%\and
%\IEEEauthorblockN{5\textsuperscript{th} Given Name Surname}
%\IEEEauthorblockA{\textit{dept. name of organization (of Aff.)} \\
%\textit{name of organization (of Aff.)}\\
%City, Country \\
%email address or ORCID}
%\and
%\IEEEauthorblockN{6\textsuperscript{th} Given Name Surname}
%\IEEEauthorblockA{\textit{dept. name of organization (of Aff.)} \\
%\textit{name of organization (of Aff.)}\\
%City, Country \\
%email address or ORCID}
}

\maketitle

\begin{abstract}
The semi-source independent quantum walk random number generator (SI-QW-QRNG) is a cryptographic protocol that extracts a string of true random bits from a quantum random walk with an adversary controls a randomness source, but the dimension of the system is known. This paper analyzes SI-QW-QRNG protocols with a memory-based quantum walk state. The new protocol utilizes a generalized coin operator with various parameters to optimize the randomness of the quantum walk state. We focus on evaluations of the protocols in multiple scenarios and walk configurations. Moreover, we show some interesting behavior of the system depending on the size of the memory space and the number of quantum coins.
\end{abstract}

%\begin{IEEEkeywords}
%Source-independent Quantum Random Number Generation
%\end{IEEEkeywords}

\section{Introduction}
A string of uniform and independent (true) random bits from random number generation has an essential role in many fields, such as cryptography, scientific simulations, game-theoretic protocols, artificial intelligence, lotteries, and basic physical tests. These works depend on the unpredictability of random numbers. Pseudo-random numbers generated via classical processes generally cannot guarantee intrinsic unpredictability. But quantum random number generation (QRNG) can produce true random numbers because of the inherent randomness throughout quantum processes. So, researching quantum random number generation (QRNG) is essential in quantum information science and engineering. Currently, cryptographically secure QRNG protocols are well-reviewed, ranging from the case of all devices utilized, sources, and measurements are fully characterized, called ``fully trusted-device'' scenario, to the case of all apparatuses operated in a QRNG protocol are not trusted, ``device-independent (DI)'' approach \cite{di-qrng0, di-qrng1, di-qrng2}. The DI-QRNG protocol is the ideal model because of its minimal assumptions for security from a cryptographic point of view. A practical experiment for the DI-QRNG protocol has been rapidly improving, but the bit-rates of such protocols cannot comply with other models \cite{di-qrng-exp,di-qrng-exp-2,si-qrng-fast}. As a midpoint, the \emph{source independent} (SI) model was introduced in \cite{vallone2014quantum} by assuming measurement devices are characterized, but the source is under the control of a dishonest party, called an adversary. In the SI model, a user wishes to generate cryptographically secure random bits without trusting the server that is probably under the adversary's control. For further information, we refer the reader to view a survey of QRNG \cite{qrng-survey
}. 
%For a survey of QRNG protocols, the reader is referred to \cite{qrng-survey}.
%It can give fast experimental bit generation rates \cite{si-qrng-fast}, and also interesting possible applications, such as the SI-model exploits the use of sunlight as a source \cite{si-qrng-sun}. 

Moreover, quantum walks (QW) are the quantum analog of classical random walks, which is an essential process in quantum computation \cite{farhi1998quantum,aharonov2001quantum,bednarska2003qwalk, childs2003exponential,childs2009universal,lovett2010universal, renato2013qwalk, proctor2014non, ashley2016qalgorithms, santha2008qwalgorithms, magniez2011qwalgorithms, balu2017qwalgorithms, kadian2021quantumwalk} and, recently, in quantum cryptography \cite{rohde2012quantum,vlachou2015quantum,vlachou2018quantum,srikara2020quantum}. Lately, a QW-based random number generation (QW-QRNG) protocol was introduced in \cite{QW-QRNG}. The first security analysis of the QW-QRNG protocol is provided to be secure in the semi-source independent (SI) model \cite{bae2021quantum}. In this paper, we extend the results from \cite{bae2021quantum} to show a more rigorous evaluation of the SI-QW-QRNG protocol with different walk parameters such as alternative coin operators and alternative models such as history-dependent quantum walks \cite{brun2003quantum, faj_2004, mcgettrick2009one, rohde2013quantum, mcgettrick2014cycle, krawec2015history}. The history-dependent quantum walk is one of the memory-based quantum walk models in which the walker exploits memory coins as it evolves. The state of the history-dependent walker contains its evolution history in the memory coins. In this paper, we employ the history-dependent quantum walk state to design a memory-based QW-QRNG protocol as the SI secure model.

We make three primary contributions to this paper. First, we develop a memory-based QW-QRNG protocol with various cases. Second, we analyze the multiple protocols from a cryptographic perspective, which produces secure systems and proves they are secured in the SI model. This analysis designates the first memory-based QW-QRNG protocol in the SI security model. Our proof utilizes the method of quantum sampling by Bouman and Fehr \cite{quantum_sampling}, expanded with techniques we developed in \cite{krawec2020new, bae2021quantum} for entropic uncertainty. Thirdly, we carefully and thoroughly evaluate the various protocols with the Hadamard, generalized, and coin-flip operators \cite{vlachou2018quantum} in the memory-based walker's evolution to optimize a random bit rate. We show some fascinating manners of the quantum random number generation procedures depending on increasing the size of the memory space and the number of quantum coins over different noises in simulated channels. For example, our protocols can improve the random bit rate against overall noises with an odd number of qubits and enhance the random bit rate of the low-dimensional position space against general noises with an even number of qubits. This odd-even behavior shows in various areas of quantum physics, such as magnetic molecular clusters, quantum dots, heavy-ion collisions, Bernal-stack trilayer graphene, the quantum Ising chain, and the quantum Szilard machine \cite{wernsdorfer1999parity,schmid2000parity,kanghun2000parity, orellana2003parity, eichler2007parity, jun2011parity, yao2012parity, beenakker2013parity, thilagam2013parity, zekun2014parity, pal2016parity, jun2019parity, petr2019parity, razmadze2020parity, parijat2020parity, congyi2021parity, yuanye2022parity, arindam2022parity, william2022parity, shaikh2022parity, sergeevich2022parity, sassetti2022parity}. Further understanding the memory-based QW-QRNG may open the door to constructing a new quantum cryptographic application, such as memory-based QW-Quantum Key Distribution (QKD) and designing a nonlocal game with an entangled quantum walk state without the freedom of will assumption that may lead to the efficient DI-QRNG and QKD with less loopholes.

\section{Preliminaries}
\subsection{Notation and Definitions}
In this section, we introduce underlying definitions and notations that will be used all over this paper. A notation $\mathcal{A}_{d}$ means a $d$-dimensional alphabet, that is, $\mathcal{A}_{d} = \{0,...,d-1\}$. So, consider a word $q\in \mathcal{A}_{d}^{N}$ and a random subset $t \subseteq [N]$, where $[N]=\{1,2,...,N\}$. Then a new word $q_{t}$ represents the substring of $q$ indexed by $t$, namely a letter in $q$ indexed by $i\in t$. A string $q_{\bar{t}}$ implies the complement of $q_{t}$ in $q$. The \emph{Hamming Weight} of $q$ is a function $wt(q) = \vert\{i : q_{i}\neq 0\}\vert$. And the \emph{relative Hamming weight} is $w(q) = wt(q)/|q|$.

A Hermitian positive semi-definite operator of unit trace that acts on a Hilbert space $\mathcal{H}$ is a \textit{density operator}. Suppose a pure quantum state $\ket{\varphi}\in\mathcal{H}$ is given. Then the \textit{density operator} of the pure state is denoted $\ketbra{\varphi}$. We simplify it as the symbol $[\varphi]$ to mean $\ketbra{\varphi}$. 

A notation $H(X)$ means the Shannon entropy of a random variable $X$. And $h_{d}(x)$ represents the $d$-ary entropy function, which is defined as: $h_{d} = x\log_{d}(d-1) - x\log_{d}(x) - (1-x)\log_{d}(1-x)$. Also, a definition $\bar{H}_{d}(x)$ is the \textit{extended $d$-ary entropy function} which is identical to $h_{d}(x)$ for all $x\in[0,1-1/d]$, is $0$ for all $x<0$, otherwise $1$ for all $x>1-1/d$.

Suppose $\rho_{BE}$ be a quantum state that acts on an arbitrary Hilbert space $\mathcal{H}_{B}\otimes \mathcal{H}_{E}$. The \textit{conditional quantum min entropy} \cite{renner2008security} is defined as follows:
$
H_{\infty}(B|E)_{\rho} = \sup_{\sigma_{E}}\max \{\lambda \in \mathbb{R}: 2^{-\lambda} I_{B}\otimes \sigma_{E} - \rho_{BE}\geq 0\},
$
where $I_{B}$ is the identity operator on $\mathcal{H}_{B}$. Consider that $\rho_{B}= \sum_{b}p_{b}[b]$ means the $E$ system is trivial and the $B$ portion is classical. So, its min entropy is $H_{\infty}(B) = -\log \max_{b}p_{b}$. The \textit{smooth conditional min entropy} \cite{renner2008security} is defined as follows:
$
H_{\infty}^{\varepsilon}(B|E)_{\rho} = \sup_{\sigma\in \Gamma_{\varepsilon}(\rho)} H_{\infty}(B|E)_{\sigma},
$
where $\Gamma_{\varepsilon}(\rho) = \{\sigma : \norm{\sigma - \rho}\leq \varepsilon\}$. Note that $\norm{X}$ is the trace distance of operator $X$.

If the $E$ portion is classical, namely the quantum-classical state $\rho_{BE} = \sum_{e}p_{e}\rho_{B}^{e}\otimes [e]$, then it can be shown that:
$
H_{\infty}(B|E)_{\rho} \geq \min_{e}H_{\infty}(B)_{\rho^{e}}.
$
Moreover, the quantum-quantum-classical state $\rho_{BER}$ has the following inequality. Let the quantum-quantum-classical state $\rho_{BER}$ be of the form $\rho_{BER} = \sum_{r}p_{r}\cdot\rho_{BE}^{r}\otimes [r]$, where the $R$ portion is classical. Then we have the following inequality:
\begin{equation}\label{eq:min_min_ent}
H_{\infty}(B|ER)_{\rho} \geq \min_{r} H_{\infty}(B|E)_{\rho^{r}},
\end{equation}
which proof is not difficult to show by using the above definitions regarding the conditional quantum min entropy.

Suppose a classic-quantum state $\rho_{BE}$ is given. Then consider $\sigma_{YE}$ is the result of a privacy amplification process on the $B$ register of the state. Through a randomly chosen two-universal hash function, the process maps the $B$ register on the $Y$ register. If the output $\ell$ bits is long, then the following relation was shown in \cite{renner2008security} that:
\begin{equation}\label{eq:true-ran-bit}
\norm{\sigma_{YE} - I_{Y}/2^{\ell}\otimes \sigma_{E}} \leq 2^{-\frac{1}{2}(H_{\infty}^{\epsilon}(B|E)_{\rho} - \ell) + 2\epsilon.}
\end{equation}

\subsection{Memory-Based Quantum Walk Models}
Several memory-based quantum walk models have been introduced, including the quantum walk with memory \cite{brun2003quantum, faj_2004, mcgettrick2009one, mcgettrick2014cycle}, the non-repeating quantum walk \cite{proctor2014non}, and the quantum walk with recycled coins \cite{rohde2013quantum}. We briefly review the quantum walk with recycled coins, called the history-dependent quantum walk (HD-QW), with additional memory space to store the coin-flip history \cite{rohde2013quantum}. The process involves a Hilbert space $\mathcal{H}_{W} = \mathcal{H}_{P}\otimes \mathcal{H}_{C_{0}}\otimes\cdots \otimes \mathcal{H}_{C_{\kappa-1}}$ where $\mathcal{H}_{P}$ is the $P$-dimensional position space, $\mathcal{H}_{C_{i}}$ is the 2-dimensional coin space, and $\kappa$ is the number of (all) recycled coins (memory coins and an active coin) in the HD-QW, where $\kappa\geq 2$ and $P\in \mathbb{N}$. Note that if $\kappa = 1$, then the quantum walk is a non-history-dependent quantum walk state. The sub-space $\mathcal{H}_{M} = \mathcal{H}_{C_{0}}\otimes \cdots \otimes \mathcal{H}_{C_{{\kappa}-2}}$ can be considered as the memory coin space that keeps track of the coin-flip history. So, the Hilbert space can be written as $\mathcal{H}_{W} = \mathcal{H}_{P}\otimes \mathcal{H}_{M}\otimes \mathcal{H}_{C},$ where $\mathcal{H}_{C} = \mathcal{H}_{c_{\kappa-1}}$, and the dimension of the walker's space $|W|$ is $2^{\kappa}\cdot P$. The walker begins at some initial position $\ket{x, c_{0},...,c_{\kappa-1}}$, e.g., $\ket{0,0,...,0}$ from which a walk operator $W$ is applied $T$ times for walker's propagation. The evolution of the walk is decomposed into three stages by adding a memory operator $M$ as follows:
\begin{align*}
C&: \ket{x, c_0,...,c_{\kappa-1}} \to \sum_{c'_{\kappa-1}}\alpha_{c'_{\kappa-1}}^{(x)}\ket{x,c_0,...,c_{{\kappa}-2}, c'_{\kappa-1}},\\
S&: \ket{x, c_0,...,c_{\kappa-1}} \to \ket{x+(-1)^{c_{\kappa-1}}, c_0,...,c_{{\kappa}-2}, c_{\kappa-1}},\\[4pt]
M&: \ket{x, c_0,...,c_{\kappa-1}} \to \ket{x, c_{\kappa-1}, c_0,...,c_{{\kappa}-2}},							
\end{align*}
where $c_i \in\{0,1\}$ and all arithmetic in the position space is computed module $P$. Let $W = M\cdot S \cdot C$ be the walk operator, where $C=I_{P} \otimes I_{M} \otimes H_{C}$. The identity matrix $I_P$ enacts on walker's position $\ket{x}$ and $I_M$ is the identity matrix on walker's memory coins $\ket{c_{0},...,c_{\kappa-2}}$. And suppose the coin operator $H_{C}$ that applies on an active coin $\ket{c_{\kappa-1}}$ is the Hadamard operator.
%\begin{equation*}
%H_{C} = \frac{1}{\sqrt{2}}\begin{bmatrix}
%       			1 & 1 \\[0.3em]
%       			1  & -1
%     		\end{bmatrix}.
%\end{equation*}
After $T$ steps, the walker evolves to state $W^{T}\ket{x,c_0,...,c_{\kappa-1}}$. At this point, a measurement may be done on the position and the coin spaces, causing a collapse at one of the $P$ and $\kappa$ spots. We denote by $\ket{w_{x,c_{\kappa}}}$ to mean the evolved state $W^{T}\ket{x,c_0,...,c_{\kappa-1}}$. We also use $\ket{w_{i}}$ when appropriate, using the natural relationship of tuples $(x,c_{0},...,c_{\kappa-1})$ to integers $i$, with $(0,0,...,0)$ being the first index $i = 0$. 

Given an HD-QW state $\ket{w_{i}}$, where $i\in\{0,...,2^{\kappa}P-1\}$, we use the notation $\mathbb{P}(\ket{w_{i}} \to (x,c_{\kappa})) = \bra{w_{i}}[x]\otimes [c_{\kappa}]\ket{w_{i}}$ to denote the probability that the walker is observed at a point of the position $x\in \{0,...,P-1\}$ and coins $c_{\kappa} = c_{0},...,c_{\kappa-1}\in\{0,1\}^{\kappa}$ after we measure the state in a position and all recycled coins. The maximum probability is defined as: $\max_{x,c_{\kappa}}\mathbb{P}(\ket{w_{i}} \to (x,c_{\kappa}))$. Since the walker's operation $W$ along with the number of steps $T$, this maximum probability can be optimized as a function of walk parameters:
\begin{equation}\label{eq:max-prob-00}
G(\kappa, P)= \min_{t}\max_{x,c_{\kappa}}\mathbb{P}(\ket{w_{i}} \to (x,c_{\kappa})),
\end{equation}
where $t\in T$. When we measure the state in a position $x$ and only memory coins $c_{\mu} = c_{0},...,c_{\kappa-2}\in \{0,1\}^{\mu}$, its probability is denoted as: $\mathbb{P}(\ket{w_{i}} \to (x,c_{\mu})) = \bra{w_{i}}[x]\otimes [c_{\mu}]\otimes I_{c_{\alpha}}\ket{w_{i}}$, where an identity operator $I_{c_{\alpha}}$ on an active coin $\ket{c_{\kappa-1}}\in \mathcal{H}_{C}$. The maximum probability function is defined as follows:
\begin{equation}\label{eq:max-prob-01}
G'(\kappa, P)= \min_{t}\max_{x,c_{\mu}}\mathbb{P}(\ket{w_{i}} \to (x,c_{\mu})),
\end{equation}
where $t\in T$. Note that when the number of recycled coins of the HD-QW state is $\kappa=1$, that is, there is no memory coin space, so the probability is defined as: $\mathbb{P}(\ket{w_{i}} \to x) = \bra{w_{i}}[x]\otimes I_{c_{\alpha}}\ket{w_{i}}$. So the function of the maximum probability is $G'(\kappa, P)= \min_{t}\max_{x}\mathbb{P}(\ket{w_{i}} \to x).$ When we only measure the state in a position $x$, then the probability is that the walker is observed at position $x$ after measurement, namely $\mathbb{P}(\ket{w_{i}}\to x) = \bra{w_{i}}[x]\otimes I_{c_{\kappa}}\ket{w_{i}}$, where $I_{c_{\kappa}}$ is an identity matrix on all recycled coins $\ket{c_{0},...,c_{\kappa-1}}\in\mathcal{H}_{M}\otimes \mathcal{H}_{C}$. The maximum probability function is defined as follows:
\begin{equation}\label{eq:max-prob-02}
G''(\kappa, P)= \min_{t}\max_{x}\mathbb{P}(\ket{w_{i}} \to x),
\end{equation}
where $t\in T$. To optimize further the function of the maximum probability, we also evaluate the protocols with the general form of the coin-rotation operator from \cite{vlachou2018quantum}, that is:,
\begin{align}\label{eq:general-coin-op}
H'_{C}(\theta, \phi) &= \begin{bmatrix}
       										  e^{i\phi}\cos(\theta) & e^{i\phi}\sin(\theta) \\[0.3em]
       									      -e^{-i\phi}\sin(\theta) & e^{-i\phi} \cos(\theta)
     			 						 \end{bmatrix},
\end{align}
where $\theta$ and $\phi$ is chosen by a user such that $\theta, \phi \in \{g\pi/ R \;|\; g = 0,1,...,R\}$ and $R\in\mathbb{Z}_{>0}$ so that the evolution operator $W = M\cdot S\cdot H'_{C}$. Moreover, a user employs a flip-coin operator $f_{c_{\alpha}}$ that is an operator acting only on the active coin $\ket{c_{\kappa-1}}$, which is to ``flip" the last coin at some initial state before evolving the walk. The flip-coin operators are as follows:
\begin{align}\label{eq:flip-coin-ops}
I = \begin{bmatrix}
       			1 & 0 \\[0.3em]
       			0 & 1
     	    \end{bmatrix}, 
X = \frac{1}{\sqrt{2}}\begin{bmatrix}
       										  1 & 1 \\[0.3em]
       									      1 & -1
     			 						   \end{bmatrix} \text{, }
Y = \frac{1}{\sqrt{2}}\begin{bmatrix}
       										  1 & 1 \\[0.3em]
       									      i & -i
     			 						   \end{bmatrix}.
\end{align}
The set of flip-coin operators is $F = \{I, X, Y\}$.
So the quantum walker state with the above descriptions is as follows: for any time $T\in \mathbb{N}$ and any initial state $\ket{x,c_0,...,c_{\kappa-1}}$, 
$\ket{w_{i}} = W^{T}(I_P\otimes I_{c_{\mu}} \otimes f_{c_{\alpha}}\ket{x,c_0,...,c_{\kappa-1}},$ where $f_{c_{\alpha}}\in F$ acts on the active coin (when the context is clear, we forgo the subscript of $f_{c_{\alpha}}$ to $f$). When the memory-based QW-QRNG employs these generalized and flip-coin operators (\ref{eq:general-coin-op}) and (\ref{eq:flip-coin-ops}), the maximum probability (\ref{eq:max-prob-00}) is redefined as follows:
\begin{equation}\label{eq:general-max-prob-00}
G(\kappa, P, F)= \min_{t,f,\theta,\phi}\max_{x,c_{\kappa}}\mathbb{P}(\ket{w_{i}} \to (x,c_{\kappa})).
\end{equation}
Also, we define the maximum probability (\ref{eq:max-prob-01}) is as follows:
\begin{equation}\label{eq:general-max-prob-01}
G'(\kappa, P, F)= \min_{t,f,\theta,\phi}\max_{x,c_{\mu}}\mathbb{P}(\ket{w_{i}} \to (x,c_{\mu})).
\end{equation}
Similar to the function (\ref{eq:max-prob-01}), when the number of coins of the state is $\kappa=1$, the function (\ref{eq:general-max-prob-01}) is defined as: $G'(\kappa, P, F) = \min_{t,f,\theta,\phi}\max_{x}\mathbb{P}(\ket{w_{i}} \to x)$. Lastly, the function of maximum probability (\ref{eq:max-prob-02}) is described as follows:
\begin{equation}\label{eq:general-max-prob-02}
G''(\kappa, P, F)= \min_{t,f,\theta,\phi}\max_{x}\mathbb{P}(\ket{w_{i}} \to x).
\end{equation}

\subsection{Quantum Sampling}
In 2010, Bouman and Fehr introduced a novel quantum sampling technique \cite{quantum_sampling}. They discovered an interesting connection relating classical sampling strategies with quantum ones, even when the quantum state is entangled with an environmental system such as an adversary. We review its concepts in this section, but we guide a reader to look through \cite{quantum_sampling} for more details. 

Suppose a string $q\in \mathcal{A}_d^N$. A classical sampling method is a procedure of selecting a random subset $t \subset [N]$ to observe $q_t$ that is the subset of the string $q$ indexed by $t$. Then it evaluates a target value of the undiscovered part, where a user can define the target function such as the Hamming weight in the unobserved portion \cite{quantum_sampling}. In the sampling strategy, we will utilize sets of selecting a subset $t$ of size $m \le N/2$ uniformly at random to have a random subset $q_t$ and an output $w(q_t)$ as an estimate of the Hamming weight in the unobserved portion. The result was presented in \cite{quantum_sampling} that, for $\delta > 0$:
\begin{equation}\label{eq:err-cl}
\varepsilon_\delta^{cl} := \max_{q\in\mathcal{A}_d^N}\mathbb{P}(q \not\in \mathcal{B}_{t,\delta}) \le 2\exp\left(\frac{-\delta^2m(n+m)}{m+n+2}\right),
\end{equation}
where the probability is over all possible selections of subsets $t$ and $\mathcal{B}_{t,\delta}$ is the group of all ``good'' words such that this sampling strategy is to almost likely produce a $\delta$-close estimate of the Hamming weight of the unobserved portion, namely:
\[
\mathcal{B}_{t,\delta} = \{q\in\mathcal{A}_{d}^{N} : |w(q_t) - w(q_{\bar{t}})| \leq \delta\}.
\]
The error probability of the classical sampling strategy is the value $\varepsilon_\delta^{cl}$, where the ``cl'' superscript means a classical sampling strategy.

In \cite{quantum_sampling}, Bouman and Fehr show the failure probabilities of the quantum sampling strategy are functions of the classical error probability. Let a basis be $\{\ket{0},..., \ket{d-1}\}$. Precisely, the choice may be arbitrary, but it is then fixed; when we use this result, the basis will be the walk basis $\{W^T\ket{x, c_{\kappa}}\}_{x,c_{\kappa}} = \{\ket{w_{i}}\}_{i}$. Define the quantum analogue of the ``good collection" of classical words as follows \cite{quantum_sampling}:
\begin{equation*}
span(\mathcal{B}_{t,\delta}) = span\{\ket{w_{i_{1}},...,w_{i_{N}}}: |w(i_{t}) - w(i_{\bar{t}})| \leq \delta\}.
\end{equation*}

Consider that if given a state $\ket{\varphi}_{AE} \in span(\mathcal{B}_{t,\delta})\otimes \mathcal{H}_E$, then a measurement in the given basis performing on those qudits indexed by $t$ leads to outcome $q\in\mathcal{A}_d^m$. Interestingly, it must keep the fact that the remaining state is a superposition of the form:
$\ket{\varphi_{t,q}} = \sum_{i\in J}\alpha_i\ket{w_{i},E_{i}},$
where $J\subset \{i \in \mathcal{A}_d^{N-m} : |w(i) - w(q)| \leq \delta\}$. The first time, Bounman and Fehr introduce the superposition lemma in the quantum sampling paper \cite{quantum_sampling}. It allows computing the lower bound of entropy of $Z$ given $E$ by using a mixed quantum state, that is, quantifying how much information an adversary $E$, prepares an entangled quantum state $\ket{\psi}_{AE}$ and sends $A$ portion to Alice, has on measuring the state by Alice. The superposition lemma is as follows:

%We have the following superposition lemma for the memory-based QW-QRNGs. It is not difficult to show the proof through using similar approaches in \cite{krawec2020new} and \cite{bae2021quantum}, but considering a quantum walker has memory and active coins. 
\begin{lemma}
%Define a history-dependent quantum walker (HD-QW) as follows:
%\begin{align*}
%\ket{w_{i}}  &= \ket{i}_{W} =W^{T}\ket{i}\\
%				&= \sum_{c_{\pi}\in \mathcal{A}_{2}^{\pi}}\sum_{c_{\kappa-1}\in\mathcal{A}_{2}}\ket{\varphi(c_{\kappa-1},i)}\otimes\underbrace{\ket{c_{0},...,c_{\pi}}}_{\ket{c_{\pi}}:\text{ memory coins}}\otimes\underbrace{\ket{c_{\kappa-1}}}_{\text{active coin}},
%\end{align*} 
%where $\pi = \kappa-2$, $\ket{i}_W\in \mathcal{H}_P\otimes\mathcal{H}_C^{\otimes \kappa}$ with $\kappa$ is the number of all (memory and active) coins, $\ket{i}$ is an initial state of HD-QW, $W$ is the unitary operator for the evolution of HD-QW over time $T\in \mathbb{Z}_{\geq 0}$, and $\ket{\varphi(c_{1},i)}\in \mathcal{H}_{P}$. 
Let $\ket{\psi}_{AE}$ and $\rho_{AE}^{\text{mix}}$ be of the form:
\begin{equation*}
\ket{\psi}_{AE} = \sum_{i\in J} \alpha_{i} \ket{i}_W \otimes \ket{E_{i}} \text{ and }
\rho_{AE}^{\text{mix}}  = \sum_{i \in J} |\alpha_{i}|^{2}[i]_{W}\otimes [E_{i}],
\end{equation*}
where $\ket{i}_W = \ket{i_1,...,i_n}_W$ is the $n$-tensors of HD-QW states, $J = \{i\in \mathcal{A}_{2^{\kappa}P}^{n}: |w(i) - w(q)|\leq \delta\}$ and $n = N-m$ with $m  = |q|$, $q\in\{0,1\}^{m}$ and $N$ is the total number of signals from the adversarial source $E$ to Alice $A$. Let $\rho_{ZE}$ and $\rho_{ZE}^{\text{mix}} =\chi_{ZE}$ describe the hybrid systems obtained by measuring subsystem $A$ of $\ket{\psi}_{AE}$ and $\rho_{AE}^{\text{mix}}$, respectively in basis $\{\ket{z}\}_{z}$ and tracing out a coin subspace out of all coin spaces, respectively. Then we have that:
\begin{equation}\label{eq:ent-super-inequality}
H_{\infty}(Z|E)_{\rho} \geq H_{\infty}(Z|E)_{\chi} - \log_{2} |J|.
\end{equation}
\end{lemma}
The quantum sampling main result \cite{quantum_sampling} translated for our application as follows:
\begin{theorem}
Let  $\delta > 0$. Given the classical sampling strategy and an arbitrary quantum state $\ket{\psi}_{AE}$, there exists a collection of ``ideal state" $\{\ket{\varphi_{t}}\}_{t}$, indexed over all possible  subsets the sampling strategy may choose, such that each $\ket{\varphi_{t}}\in span(\mathcal{B}_{t,q})\otimes \mathcal{H}_E$ and:
\begin{equation}
\frac{1}{2}\norm{\frac{1}{T}\sum_{t} [t]\otimes [\psi] - \frac{1}{T}\sum_{t}[t]\otimes [\varphi_{t}] } \leq \sqrt{\varepsilon_{\delta}^{cl}},
\end{equation}
where $T = {N \choose m}$ and the sum is over all subsets of size $m$.
\end{theorem}

\section{Our Protocol}
We build a memory-based QW-QRNG protocol based on the QW-QRNG protocol introduced in \cite{bae2021quantum}. Our protocol is to generate a secure string of true random bits by using an HD-QW state \cite{rohde2013quantum}. Without loss of generality, the randomness source, possibly controlled by an adversary, prepares the $N$-walkers and sends them to Alice. If a source is honest, then it should prepare the state $\ket{w_{0}}^{\otimes N} = \big(W^{T}\ket{0,0,...,0}\big)^{\otimes N}$, independent of $\mathcal{H}_{E}$ and send it to Alice, but the adversarial source may prepare anything. Also, there are no assumptions on this state's overall structure beyond that it consists of the $N$-walkers, and it may even be non-i.i.d. The goal of the memory-based QW-QRNG is to produce a uniformly random string, independent of any adversary's system. The general memory-based QW-QRNG protocol is as follows:\\
\textbf{Public Parameters:} The quantum walk setting includes the dimension of the position space $P$, the dimension of the coin space $c_i$, the $\kappa$-number of (all) recycled coins, the walker's unitary operator $W$, and the number of steps to evolve by $T$. The walker exploits the $\mu=\kappa -1$ number of memory coins $\ket{c_{\mu}}=\ket{c_{0},...,c_{\kappa - 2}}$ as a memory of the coin evolution. When a user employs the generalized coin operator (\ref{eq:general-coin-op}), an angle $\theta$ and a phase $\phi$ are public parameters, otherwise the Hadamard coin operator used. Also, if the protocol exploits the set of flip-coin operators, $F$, (\ref{eq:flip-coin-ops}) that each operator flips the (last) active coin at some walker's initial state before the evolution, the set $F$ is to be known to the public.\\
\textbf{Source:} An untrusted source, possibly adversarial, produces an HD-QW state. $\ket{\psi_{0}} \in \mathcal{H}_{A} \otimes \mathcal{H}_{E},$ where $\mathcal{H}_{A} \cong \mathcal{H}_{W}^{\otimes N}$. If the source is honest, the state prepared should be of the form $\ket{\psi_{0}} = \ket{w_{0}}^{\otimes N} \otimes \ket{0}_{E},$ namely, $N$-copies of the walker state $\ket{w_{0}} = W^{T}\big(\ket{0}\otimes \ket{0,...,0}\big)$ unentangled  with Eve $E$. Note that Alice can define which state is a honest state $\ket{w_{0}}$.\\
\textbf{User:} Alice chooses a random subset $t \subset [N]$ of size $m$ and measures the systems indexed by this subset using POVM $\mathcal{W} = \big\{[w_{0}], I-[w_{0}]\big\} = \{W_{0}, W_{1}\}$ resulting in outcomes $q\in\{0,1\}^{m}$. She measures the remaining $n$-walker systems in a basis of measurements $\mathcal{Z} $, where $n=N-m$. The first outcome is used to test the fidelity of the received state while the second is used as a raw-random string $r\in \mathcal{A}_{|\mathcal{Z}|}^{n}$. \\
%The protocol rebuilds each random digit into many random bits—finally, the string $r$ changes to the series of random bits $\tilde{r}$.\\
\textbf{Postprocessing: } Alice applies privacy amplification to $r$, producing a final random string of size $\ell$. As proven in \cite{frauchiger2013true}, the hash function used for privacy amplification need only be chosen randomly once and then reused for each run of the protocol for a QRNG protocol of this nature.

There are several cases of considered protocols for the memory-based QW-QRNG as follows. For each case, most settings are similar to the general protocol above, but a set of measurement operators $\mathcal{Z}$ in a randomness extraction mode is differently defined as follows.
\subsection{The Using All Case: Using Memory and Active Coins}
The using all case that the memory-based QW-QRNG protocol utilizes all memory and active coins to produce a string of true random bits. A measurement in an extraction mode is as follows. After testing the fidelity, a user, Alice, measures the remaining $n$-walker systems in orthonormal (computational) bases of $\mathcal{H}_{2^{\kappa}P}$. The set of measurement operators is defined as follows: for all orthonormal bases $\ket{z_{i}}\in\mathcal{H}_{W}$,
\begin{equation}
\mathcal{Z} = \big\{[z_{i}]\big\}_{i=0}^{d-1} = \big\{Z_{i}\big\}_{i=0}^{d-1},
\end{equation}
where $d=2^{\kappa}P$. 
%So this measurement produces the $2^{\kappa}P$-number of outcomes. 
\subsection{The Using Memory Case: Only Using Memory Coins}
The protocol that generates a string of random bits with only using the $\mu$ number of memory coins $\ket{c_{\mu}}=\ket{c_{0},...,c_{\kappa-2}}$ to generate a string of true random bits. After the testing the fidelity, a user, Alice, measures the remaining $n$-walker systems in a following POVM measurement:
\begin{equation}\label{eq:povm-z-01}
\mathcal{Z}' = \big\{[j]\otimes [c_{\mu}]\otimes I_{c_{\alpha}}\big\}_{j, c_{\mu}} = \big\{ Z_{j,c_{\mu}}\big\}_{j, c_{\mu}} = \{Z'_{i}\}_{i=0}^{d-1},
\end{equation}
where $j=0,...,P-1$, $c_{\mu} = c_{0},...,c_{\kappa-2} \in \{0,1\}^{\mu}$, $I_{c_{\alpha}}$ is an identity matrix on the active coin $\ket{c_{k-1}}$, and $d = 2^{\mu}P$.

\subsection{The Not Using Memory Case: Not Using Memory and Active Coins}

The protocol that generates a string of random bits without using any coins of each quantum walker to generate true random bits sequence. After the testing the fidelity, a user, Alice, measures the remaining $n$-walker systems in a following POVM measurement:
\begin{equation}\label{eq:povm-z-02}
\mathcal{Z}'' = \big\{[j]\otimes I_{c_{\mu}}\otimes I_{c_{\alpha}}\big\}_{j=0}^{P-1} = \{Z_{j}\}_{j=0}^{P-1} = \{Z''_{i}\}_{i=0}^{d-1}, 
\end{equation}
where $I_{c_{\mu}}$ is the identity matrix on the memory coins, $I_{c_{\alpha}}$ is the identity matrix on the active coin, and $d = P$.

For all cases, after the extraction, the post-processing to produce true random bits is as same as the above description. The security of the using all case protocol can be proven using the original sampling-based entropic uncertainty relations \cite{krawec2020new}. But the new sampling-based entropic uncertainty relations stemmed from \cite{bae2021quantum} can prove the security of the using and non-using memory cases of the protocols.

\section{Security Analysis}

This section mainly is to show how the protocols of the memory-based QW-QRNG produce a secure string of true random bits. It requires a bound on the quantum min-entropy, using Eq. (\ref{eq:true-ran-bit}), from the number of random bits $\ell$ that may be extracted from the $N$-numbers of HD-QW states prepared by an adversarial source.  Without loss of generality, the adversarial source is empowered to create any initial state, possibly entangled with her ancilla. However as in \cite{vallone2014quantum}, the dimension of the system sent to Alice is known; in our case it is $(2^{\kappa}P)^{N}$, that is, $N$ quantum walker states, each of dimension $2^{\kappa}P$. We consider that each of the $N$-walkers is a possibly different state, which scenario models natural noise and considers an adversarial source. Lastly, Alice's measurement devices are fully-characterized, meaning that her measurement devices are in the trust.

\subsection{The Using All Case}
For the using all case, we can show the security trough the result in \cite{krawec2020new}. The goal is to show the protocol's produced random string is uniformly random and independent of any adversary's system. We can evaluate how many uniform independent random numbers the using all protocol extracts from the HD-QW states using the min-entropy. To evaluate the min-entropy, we exploit Theorem 2 in \cite{krawec2020new} that is as follows:

%In particular, through a privacy amplification process, taking a classical-quantum (cq) state $\rho_{ZE}$ and process the $Z$ register that is the $n$-length long to transform into the cq-state $\sigma_{YE}$, where $Y$ register is the $\ell$-length long by hashing it through a two-universal hash function, the following result \cite{renner2008security} holds that:
%\begin{align}
%\norm{\sigma_{YE} - \frac{I_{Y}}{|Y|}\otimes \sigma_{E}} \leq 2^{-\frac{1}{2}(H_{\infty}^{\varepsilon}(Z|E)_{\rho} - \ell )} + 2\varepsilon = \varepsilon_{PA}.
%\end{align}
%Let $\varepsilon >0$ and set $\varepsilon_{PA} = 9\varepsilon + 4\varepsilon^{\beta}$ be the desired distance from an ideal uniform random string of size $\ell$-independent of $E$'s system. By simplification of the equation $(*)$, we have the following length of random numbers:
%\begin{align}
%\ell  &= H_{\infty}^{\varepsilon}(Z|E)_{\rho} -2\log{\bigg(\frac{1}{\varepsilon_{PA} -  2\varepsilon}\bigg)}.
%\end{align}
%We need to estimate the smooth min-entropy of $\rho_{ZE}$, which is hard to determine. But, by using the result, Theorem 2 in \cite{krawec2020new}, we can evaluate the min-entropy. The result, using the quantum sampling and the entropic uncertainty principle, is as follows:

\begin{theorem}
Let $\varepsilon > 0$, $0<\beta < 1/2$, and $\rho_{AE}$ an arbitrary quantum state acting on $\mathcal{H}_{A}\otimes \mathcal{H}_{E}$, where $\mathcal{H}_{A}\cong \mathcal{H}_{d}^{\otimes(n+m)}$ for $d\geq 2$ and $m<n$. Let $Z =\{\ket{z_{i}}\}_{i=0}^{d-1}$ and $X =\{\ket{x_{i}}\}_{i=0}^{d-1}$ be two orthonormal bases of $\mathcal{H}_{d}$ and $\Lambda$ be two outcome POVM with elements $\{\Lambda_{0} = [x_{0}], \Lambda_{1} = I - [x_{0}]\}$. If a subset $t$ of size $m$ of $\rho_{A}$ is measured using $\Lambda$ resulting in outcome $q$ we denote by $\rho(t,q)$ to be the post-measurement state. Then it holds that:
\begin{equation}
\mathbb{P}\bigg(H_{\infty}^{\varepsilon'}(Z|E)_{\rho(t,q)}+\frac{n\bar{H}_{d}(w(q)+\delta)}{\log_{d}(2)}  \geq n\gamma\bigg) \geq 1 - \varepsilon'',
\end{equation}  
where the probability is over the choice of subset $t$ and the measurement outcome $q$. Above:
\begin{equation}\label{eq:max-overlap}
\gamma = -\log_{2}\max_{a,b\in\mathcal{A}_{d}}|\braket{z_{a}}{x_b}|^{2},
\end{equation}
and $\varepsilon' = 4\varepsilon + 2\varepsilon^{\beta}$, $\varepsilon'' = 2\varepsilon^{1-2\beta}$ and finally:
\begin{equation*}
\delta = \sqrt{\frac{(m+n+2)\ln (2/\varepsilon^{2})}{m(m+n)}}.
\end{equation*}
\end{theorem}

%The extended d-ary entropy function, denoted $\bar{H}_{d}(x)$, for any $x\in \mathbb{R}$, is defined as follows:
%\begin{align*}
%\bar{H}_{d}(x)=
%  \begin{cases}
%     0       & \quad \text{if } x < 0\\
%    h_{d}(x)  & \quad \text{if } 0 \leq x \leq 1 - 1/d\\
%   1  & \quad \text{if } x > 1 - 1/d,
%  \end{cases}
%\end{align*}
%where the d-ary entropy function is defined to be:
%\begin{align*}
%h_{d}(x) = x\log_{d}(d-1) - x\log_{d}(x) - (1-x)\log_{d}(1-x).
%\end{align*} 

In the using all case, the square of the maximum overlap in Eq. (\ref{eq:max-overlap}) is equivalent to $\max_{x,c_{\kappa}}\mathbb{P}(\ket{w_{i}}\to (x,c_{\kappa}))$. So, the gamma (\ref{eq:max-overlap}) can be defined by the function of the maximum probability (\ref{eq:max-prob-00}) as follows: given an honest HD-QW state $\ket{w_{0}}$,
\begin{equation}\label{eq:gamma-00}
\gamma = -\log_{2} G(\kappa, P).
\end{equation}
When the using all case employ the generalized and flip coin operators, the gamma (\ref{eq:max-overlap}) is defined by the maximum probability (\ref{eq:general-max-prob-00}). Then, except with the failure probability $2\varepsilon^{1-2\beta}$, we find the lower bound of the size of the random string as follows:
\begin{equation}
\ell_{ours} \geq n\Bigg(\gamma - \frac{\bar{H}_{|W|}(w(q)+\delta)}{\log_{|W|}(2)}\Bigg) - 2\log{\frac{1}{\tilde{\varepsilon}}}
\end{equation}
where $\tilde{\varepsilon} = \varepsilon_{PA} -  2\varepsilon$. A user generally runs the protocol and observes $q$ directly. But, we consider the noise follows a depolarization channel with parameter $Q$ to simulate its implementation. This noise model is a standard one to estimate in simulations. After sampling, Alice will have an expected Hamming weight in her test measurement of $w(q) = Q$. 

%To assess the using all protocol $(\ell_{ours})$, we set $\beta = 1/3$ and $\varepsilon = 10^{-36}$ which implies the failure probability is $\varepsilon'' = 2\varepsilon^{1-2\beta} = 2\times10^{-12}$ so that $\varepsilon_{PA} = (9\times 10^{-3} +4)10^{-12} \cong 4\times10^{-12}$. Note that we did not optimize $\beta$, which may lead to higher rates for the protocol, and we use $7\%$ of total signals for sampling.

\subsection{The Using Memory and Not Using Memory Cases}
We follow the security analysis of the QW-QRNG protocol in \cite{bae2021quantum} to show the security of the using memory and the not using memory cases in the memory-based QW-QRNG protocol with HD-QW states. We amended Theorem 2. in \cite{bae2021quantum} for the using memory and not using memory cases. For both cases, the following main theorem exploits different post-measurement states. The details of the theorem is as follows:

%Similarly, standard entropic uncertainty relations of the form \cite{tomamichel2011uncertainty}:
%\begin{equation}\label{eq:ent}
%\Hmin^\epsilon(A|E) + \Hmax^\epsilon(A|B) \ge -\log\max_{x,y}\norm{\sqrt{W_x}\sqrt{Z_y}}^2,
%\end{equation}
%where the norm is the operator norm and two POVMs $\mathcal{W} = \{W_{x}\}$ and $\mathcal{Z} = \{Z_{y}\}$ are used in both of the using and not using memory cases, are not fitting and can only produce the trivial bound. 
\begin{theorem}
Let $\varepsilon > 0$. After executing the protocol of the using memory case and observing outcome $q$ during the test stage (namely, after measuring using $\mathcal{W}$), it holds that, except with probability at most $\varepsilon^{1/3}$ (where the probability here is over the choice of sample subset and observation $q$), the protocol outputs a final secret string of size:
\begin{equation}
\ell^{'}_{ours} = \eta_{q}\gamma' - n\cdot \frac{\bar{H}_{2^{\kappa}P}(w(q)+\delta)}{\log_{2^{\kappa}P}(2)} -  2 \log_{2}\frac{1}{\varepsilon}, %- \log_{2} {N \choose m}
\end{equation}
which is $(5\varepsilon  + 2\varepsilon^{1/3})$-close to an ideal random string (i.e., one that is uniformly generated and independent of  any  adversary system as in Eq. (\ref{eq:true-ran-bit})). Above, given an honest HD-QW state $\ket{w_{0}}$, the gamma $\gamma'$ of the using memory case with the maximum probability (\ref{eq:max-prob-01}) defines as follows: 
\begin{equation}\label{eq:gamma-01}
\gamma' = -\log_{2} G'(\kappa, P),
\end{equation}
the gamma $\gamma'$ of the not using memory case with the maximum probability (\ref{eq:max-prob-02}) defines as follows: 
\begin{equation}\label{eq:gamma-02}
\gamma' = -\log_{2} G''(\kappa, P),
\end{equation}
and 
$\eta_{q} = (N-m)(1-w(q)-\delta)$, where $\delta$ is:
\begin{equation}
\delta = \sqrt{\frac{(N+2)\ln (2/\varepsilon^{2})}{m\cdot N}}.
\end{equation}
Note that when the using memory and the not using memory cases exploit the generalized and flip coin operators (\ref{eq:general-coin-op}) and (\ref{eq:flip-coin-ops}), the gammas (\ref{eq:gamma-01}) and (\ref{eq:gamma-02}) employ the maximum probabilities (\ref{eq:general-max-prob-01}) and (\ref{eq:general-max-prob-02}), respectively.
\end{theorem}

\begin{proof}
%Using Theorem 1. (with  respect to the reference basis $\{W^{T}\ket{0\cdots 0,0},...,W^{T}\ket{1\cdots 1, P-1}\}$),  
Suppose that $\ket{\psi_{i}}_{AE}$ be the quantum state that the dishonest source $E$ creates and sends the $A$ portion to Alice $N$-times with the fixed error $\varepsilon  > 0$. If the source is honest, Alice receives the state $\ket{\psi_{0}}_{AE} = \ket{w_{0}}^{\otimes N}\otimes \ket{E_{0}}$, where $\ket{w_{i}} \in \mathcal{H}_{W} = \mathcal{H}_{P}\otimes \mathcal{H}_{M}\otimes \mathcal{H}_{C}$ and $i \in \{0,...,2^{\kappa}P - 1\}$. Note that $\kappa$ is the number of all (memory and active) coins and $P$ is the walker's positional dimension. By Theorem 1. in the quantum sampling, there exists ideal states, indexed over all subsets $t\subset [N]$ of size $m$ such that $\ket{\varphi_{t}}\in span\big(\ket{w_{i_1},...,w_{i_N}}: |w(i_{t}) - w(i_{\bar{t}})| \leq \delta\big)\otimes \mathcal{H}_E$. Note that we define $\ket{w_{0}} = \ket{0}_{P}\otimes\ket{0\cdots 0}_{c_{\mu}}\otimes \ket{0}_{c_{\alpha}}$, where $c_{\mu} = c_{0},...,c_{\kappa-2}$ and $c_{\alpha} = c_{\kappa-1}$. We utilize a similar approach, a two-step proof method, from \cite{krawec2020new}, \cite{bae2021quantum} to show the security. Analyzing the security of the ideal state $\sigma_{TAE} = 1/T\sum_{t}[t] \otimes [\varphi_{t}],$ where $T = {N \choose m}$ and $p(t) = 1/T$ is the first footstep. We measure the $T$ register in $\sigma_{TAE}$, which leads the state to collapse to the superposition of ideal states $\ket{\varphi_{t}}$ that is a quantum analogy of a classical random sampling. After the quantum sampling, we measure the ideal states $\ket{\varphi_{t}}$ in a set of measurement POVM $\mathcal{W}$ to have $q\in \{0,1\}^{m}$.  It tests whether a sample of quantum states $\ket{w_{i}}$ is honest or not. If the outcome is $q=0$, i.e., $\ket{w_{0}}$, Alice considers the state from the source is honest so that she can extract true randomness. Then the experiment traces out the measured portion of size $m$ resulting in the post-measurement state $\sigma(t,q)$ acting on $\mathcal{H}_{2^{\kappa}P}^{\otimes n} \otimes \mathcal{H}_{E}$. Since $\ket{\varphi_{t}}\in span(\mathcal{B}_{t,2^{\kappa}P}^{\delta})\otimes \mathcal{H}_{E}$, we claim that the post-measurement state is of the form:
\begin{equation}
\sigma(t,q) = \sum_{e\in\mathcal{A}_{(2^{\kappa}P) - 1}^{wt(q)}} p_{e}\cdot \sigma_{AE}^{(e)},
\end{equation}
where $P(Z) = ZZ^{*}$, $wt(q)$ is the (non-relative) Hamming weight of $q$, $\sigma_{AE}^{(e)}= P\big(\sum_{i\in J_{q}^{(e)}} \alpha_{i}^{(e)} \ket{w_{i}} \otimes \ket{E_{i}}\big)$, and $J_{q}^{(e)} \subset \{i\in \mathcal{A}_{2^{k}P}^{n}: |w(i) - w(q)|\leq \delta\}.$ Recall $n=N-m$. This is the form of the post-measurement state after the experiment is done. Indeed, note that $\ket{\varphi_{t}}$ is a superposition of vectors of the form $\{\ket{w_{i}}: |w(i) - w(q)|\leq \delta \}$. Thus, on observing $q$ using POVM $\mathcal{W}$ on subspace indexed by $t$, but before tracing out the measured portion, the state is of the form:
\begin{equation}\label{eq:post-measure-01}
\sum_{e\in \Omega_q}\sqrt{p_e}\ket{x_e}_Q \otimes \sum_{i\in J_{q}^{(e)}}\alpha_{i}^{(e)}\ket{w_{i}}\otimes \ket{E_{i}},
\end{equation}
where $\Omega_q = \{e\in \mathcal{A}_{2^{k}P}^{m}: e_{i} = 0 \text{ iff } q_{i}  = 0\}$. As the final step of the experiment, tracing out the $Q$ register brings Eq. (\ref{eq:post-measure-01}). Let us consider one of the $\sigma_{AE}^{(e)}$ states. And in the using memory case, Alice performs a measurement using POVM $\mathcal{Z}'$ (\ref{eq:povm-z-01}), on the remaining $A$ portion to extract true randomness based on the hi-dimensional space (in the not using memory case, she measures the state using POVM $\mathcal{Z}''$ (\ref{eq:povm-z-02}). For the using memory case, to compute this state $\sigma(t,q)$, we write a single quantum walker $\ket{w_{i}}\in \mathcal{H}_W$ as $\ket{w_{i}} = \sum_{c_{\alpha}}\ket{\varphi(c_{\alpha},i)}\otimes \ket{c_{\alpha}}$, where $\ket{\varphi(c_{\alpha}, i)}$ are walker's states in $\mathcal{H}_{2^{\mu}P}$ and $\ket{c_{\alpha}}$ is the active coin. For the not using memory case, $\ket{w_{i}}\in \mathcal{H}_W$ as $\ket{w_{i}} = \sum_{c_{\kappa}}\ket{\varphi(c_{\kappa},i)}\otimes \ket{c_{\kappa}}$, where $\ket{c_{\kappa}}$ are all memory and active coins and $\ket{\varphi(c_{\kappa}, i)}$ are walker's states in $\mathcal{H}_{P}$. With this notation, in the using memory case, we can compute a post-measurement state, with Alice storing the outcome $z\in \mathcal{A}_{2^{\mu}P}^{n}$ in a classical register $Z$ and also tracing out the active coin register. The post-measurement state is as follows:
\begin{equation*}
\sigma_{ZE}^{(e)} = \sum_{z}[z]_{Z}\sum_{i,j\in J_{q}^{(e)}}\alpha_{i}\alpha_{j}^{*}\sum_{c_{\alpha}\in\{0,1\}^{n}} \beta_{z,c_{\alpha},i}\beta_{z,c_{\alpha},i}^{*}\otimes [E_{ij}],
\end{equation*}
where $z\in\mathcal{A}_{2^{\mu}P}^{n}$ and $\beta_{z,c_{\alpha},i} = \prod_{\ell=0}^{n-1}\braket{z_{\ell}}{\varphi(c_{\alpha}^{(\ell)},i_{\ell})}$. In the not using memory case, Alice saves the outcome $z\in\mathcal{A}_{P}^{n}$ in a classical register $Z$ after her measurements with the set $\mathcal{Z}''$ and tracing out the all memory and active coins register. Then the post-measurement state is as follows:
\begin{equation*}
\sigma_{ZE}^{(e)} = \sum_{z}[z]_{Z}\sum_{i,j\in J_{q}^{(e)}}\alpha_{i}\alpha_{j}^{*}\sum_{c_{\kappa}\in\{0,1\}^{\kappa n}}\beta_{z,c_{\kappa},i}\beta_{z,c_{\kappa},i}^{*}\otimes [E_{ij}],
\end{equation*}
where $z\in\mathcal{A}_{P}^{n}$ and $\beta_{z,c_{\kappa},i} = \prod_{\ell=0}^{n-1}\braket{z_{\ell}}{\varphi(c_{\kappa}^{(\ell)},i_{\ell})}$. In the using memory case, to compute the min-entropy of the post-measurement state, we will consider the following density operator:
\begin{equation*}
\chi_{ZE} = \sum_{z}[z]\sum_{i}|\alpha_{i}|^{2}\sum_{c_{\alpha}} |\beta_{z,c_{\alpha},i}|^{2} \otimes [E_i].
\end{equation*} 
Similarly, we will have the following density operator in the not using memory case:
\begin{equation*}
\chi_{ZE} = \sum_{z}[z]\sum_{i}|\alpha_{i}|^{2}\sum_{c_{\kappa}}|\beta_{z,c_{\kappa},i}|^{2}\otimes [E_{i}].
\end{equation*}

Due to the superposition lemma $(12)$, in both cases, it bounds the min-entropy of a superposition based on the min-entropy of both suitable mixed states, we find that:

\begin{equation}
H_{\infty}(Z|E)_{\sigma^{(e)}} \geq H_{\infty}(Z|E)_{\chi} - \log_2 |J_{q}^{(e)}|.
\end{equation}

Consider the state $\chi_{ZEI}$ where we append an auxiliary system spanned by orthonormal basis $\ket{i}$ as follows:
\begin{equation}
\chi_{ZEI} = \sum_{i}|\alpha_{i}|^{2}\cdot\chi^{(i)}\otimes [E_{i}] \otimes [i],
\end{equation}
where $\chi^{(i)} = \sum_{z}[z]\sum_{c_{\alpha}}|\beta_{z,c_{\alpha},i}|^{2}$ in the using memory case and $\chi^{(i)}=\sum_{z}[z]\sum_{c_{\kappa}}|\beta_{z,c_{\kappa},i}|^{2}$ in the not using memory case. Particularly, in the using memory case, for strings $z\in \mathcal{A}_{2^{\mu}P}^{n}$ and $i\in \mathcal{A}_{2^{k}P}^{n}$, let $p(z|w_{i})$ be the probability that outcomes $z$ is observed if measuring the pure, and unentangled state, state $\ket{w_{i_1},...,w_{i_n}}$ using POVM $\mathcal{Z}'$. Simple algebra shows that this is in fact $p(z|w_{i}) = \sum_{c_{\alpha}}|\beta_{z,c_{\alpha},i}|^{2}$. So $\chi^{(i)} = \sum_{z}p(z|w_{i})[z]$. Similarly, in the not using memory case, for strings $z\in \mathcal{A}_{P}^{n}$ and $i\in \mathcal{A}_{2^{k}P}^{n}$, let $p(z|w_{i})$ be the probability that outcomes $z$ is observed if measuring the pure, and unentangled state, state $\ket{w_{i_1},...,w_{i_n}}$ using POVM $\mathcal{Z}''$. Again, simple algebra shows that is in fact $p(z|w_{i}) = \sum_{c_{\kappa}}|\beta_{z,c_{\kappa},i}|^{2}$. So $\chi^{(i)} = \sum_{z}p(z|w_{i})[z]$. From the strong subadditivity of the min-entropy \cite{renner2008security}, in both cases, from Equation (\ref{eq:min_min_ent}), and treating the joint $EI$ registers as a single classical register, we have:
\begin{equation}\label{eq:subadditiviity-01}
H_{\infty}(Z|E)_{\chi} \geq H_{\infty}(Z|EI)_{\chi}  \geq \min_{i} H_{\infty}(Z)_{\chi^{(i)}}.
\end{equation}
Fix a particular $i \in J_{q}^{(e)}$ and $\eta = n - wt(i)$ (namely, $\eta$ is the number of zeros in the string $i$.) Then, in the using memory case, it is clear that $p(z|w_{i})\leq \max_{x, c_{\mu}} \mathbb{P}_W (\ket{w_{0}}\to (x, c_{\mu}))^{\eta} = \xi^{\eta}$. Indeed, any other  $\mathbb{P}_W (\ket{w_{0}}\to (x, c_{\mu})) \leq 1$. Similarly, in the not using memory case, it shows that $p(z|w_{i})\leq \max_{x} \mathbb{P}_W (\ket{w_{0}}\to x)^{\eta} = \xi^{\eta}$, where $\mathbb{P}_W (\ket{w_{0}}\to x) \leq 1$. So in the both cases, we may consider only the $\ket{w_{0}}$ term as contributing to this upper-bound. From this it follows that $H_{\infty}(Z)_{\chi^{(i)}} = -\log \max_{z}p(z|w_{i}) \geq -\log \xi^{n-wt(i)}$. From \cite{krawec2020new} and \cite{bae2021quantum}, by considering the noise in the source via the sampling, we have that for $i\in J_{q}^{(e)}$ and $w(i)\leq n(w(q)+\delta)$, $\min_{i} H_{\infty} (Z)_{\chi^{(i)}} \geq -\log \xi^{n(1-w(q)\delta)} \geq \eta_{q}\gamma',$
where $\eta_{q} = n(1-(w(q)+\delta))$, the gamma is $\gamma'$ (\ref{eq:gamma-01}) in the using memory case, and the gamma is $\gamma'$ (\ref{eq:gamma-02}) in the not using memory case. Finally, by the well known bound on the volume of a Hamming ball $|J_{q}^{(e)}| \leq d^{n\bar{H}_{d}(w(q)+\delta)}$, Eq. (\ref{eq:ent-super-inequality}) in the superposition lemma, and the inequality (\ref{eq:subadditiviity-01}), we have the bound of the min-entropy:
%\begin{align*}
%H_{\infty}(Z|E)_{\sigma^{(e)}} &\geq H_{\infty}(Z|E)_{\chi} - \log_2 |J_{q}^{(e)}|\\
%&\geq \min_{i} H_{\infty}(Z)_{\chi^{(i)}} - \frac{n\bar{H}_{2^{\kappa}P}(w(q)+\delta)}{\log_{2^{\kappa}P}(2)}\\
%&\geq \eta_{q}\gamma' - \frac{n\bar{H}_{2^{\kappa}P}(w(q)+\delta)}{\log_{2^{\kappa}P}(2)}.
%\end{align*}
\begin{equation}
H_{\infty}(Z|E)_{\sigma^{(e)}} \geq \eta_{q}\gamma' - \frac{n\bar{H}_{2^{\kappa}P}(w(q)+\delta)}{\log_{2^{\kappa}P}(2)}.
\end{equation}

The above analysis is for the ideal state. Since we use a similar technique that we employed in \cite{krawec2020new} and \cite{bae2021quantum} for translating this ideal analysis to the real case, we refer the details to \cite{krawec2020new} and \cite{bae2021quantum}.
%Indeed, let $\rho(t,q)$ be the state of the real system, $\ket{\psi}_{AE}$, conditioned on the protocol sampling subset $t$ and observing outcome $q$ and let $\sigma(t,q)$ be the same for the ideal state.  If we define:
%$\Delta_{t,q} = \frac{1}{2}\norm{\rho(t,q) - \sigma(t,q)},$ then, treating $\Delta_{t,q}$ as a random variable over the choice of $t$ and outcome $q$, it can be shown (see the proof of Theorem 2 in \cite{krawec2020new} for explicit details) that except with probability $\epsilon^{1/3}$, it holds that $\Delta_{t,q} \leq \epsilon + \epsilon^{1/3}$ where the probability is over the choice of $t$ and outcome $q$. Thus, by switching to smooth min entropy, we have, except with probability at most $\epsilon^{1/3}$ that $H_{\infty}^{4\epsilon+2\epsilon^{1/3}}(Z|E)_\rho \geq H_{\infty}(Z|E)_\sigma$. Privacy amplification (Equation $(2)$, setting the right-hand-side of that equation equal to $\epsilon_{PA} = 5\epsilon+2\epsilon^{1/3}$, namely twice the smoothening parameter plus an additional $\epsilon$) completes the proof.
\end{proof}

%\subsection{The Not Using Memory Case}
%
%\begin{theorem}
%Let $\varepsilon > 0$. After executing the protocol of the not using memory case and observing outcome $q$ during the test stage (namely, after measuring using $\mathcal{W}$), it holds that, except with probability at most $\varepsilon^{1/3}$ (where the probability here is over the choice of sample subset and observation $q$), the protocol outputs a final secret string of size:
%\begin{equation}
%\ell^{''}_{ours} = \eta_{q}\gamma'' - n\cdot \frac{\bar{H}_{2^{\kappa}P}(w(q)+\delta)}{\log_{2^{\kappa}P}(2)} -  2 \log_{2}\frac{1}{\varepsilon}, 
%\end{equation}
%which is $(5\varepsilon  + 2\varepsilon^{1/3})$-close to an ideal random string (i.e., one that is uniformly generated and independent of  any  adversary system as in Equation $(5)$). Above, the gamma $\gamma''$ defines as follows:
%$
%\gamma'' =  -\log_{2} \max_{z}\mathbb{P}_{W}\big(\ket{w_{0,0,...,0}}\to z\big),
%$
%and 
%$\eta_{q} = (N-m)(1-w(q)-\delta)$, where the delta $\delta$ is:
%\begin{equation}
%\delta = \sqrt{\frac{(N+2)\ln (2/\varepsilon^{2})}{m\cdot N}}.
%\end{equation}
%\end{theorem}
%The proof of Theorem 4 is similar to the argument used to establish Theorem 3, so we omit the details. 

\section{Evaluation}
Computing the gamma $\gamma$ is crucial to produce a string of true random bits. The gamma can be optimized because the walker's evolution operator can be parameterized by time $t$, an angle $\theta$, a phase $\phi$, and flip-coin operators throughout the protocols. We will take a look at various scenarios for the different cases. First, we evaluate the performance of the above protocols under a variety of walker's dimensions $|W| = 2^{\kappa}\cdot P$ with different positions $P = 3, 5, 11, 21 \text{, and } 51$, different the number of recycled coins $\kappa = 1, 2, 3 \text{,  and }4$, and use of the Hadamard coin operator, the generalized, and flip coin operators (\ref{eq:general-coin-op}) and (\ref{eq:flip-coin-ops}). To optimize $H_{\infty}^{\varepsilon}(Z|E)_{\rho}$ as so the random bit rate, it needs the utmost gamma $\gamma$ over time $t$ for fixed dimension P and the number of coins $\kappa$. We found that for fixed position $P$ and the number of coins $\kappa$, the maximal gamma $\gamma$ value over all time setting $t=1,...,2000$ with the Hadamard coin operator and $t=1,...,1000$ with the generalized and flip coin operators. We set that the source sends the $N$-number of signals. Then a user employs a sample size that is the square root of the total number of signals $N$, namely, $m=\sqrt{N}$. The user computes a random bit rate $\ell_{ours}/N$.

\subsection{The Using All Case}
In the using all case, the protocol utilizes the memory and active coins to generate a string of true random bits. The secure random bit rates are computed through the main theorems from \cite{krawec2020new}, \cite{bae2021quantum} and their modifications for the version of the HD-QW state\cite{rohde2013quantum}. To compute the gamma $\gamma$ in Eq. (\ref{eq:gamma-00}), we employ the maximum probability with the Hadamard operator. Note that for $\kappa = 1$, which is the non-memory based quantum walk case, the case has the following maximum probabilities: $G(1,3) = 0.2224$, $G(1,5) = 0.1474$, $G(1,11) = 0.0983$, $G(1,21) = 0.0642$, and $G(1,51) = 0.0367$ over $T$. For $\kappa >1$, the using all case with the maximum probability (\ref{eq:max-prob-00}) and the maximum probability (\ref{eq:general-max-prob-00}) are evaluated in table I and II, respectively. Note that in the table II, the function $G(\kappa, P, F)$ is written as $G(\kappa, P)$.

\begin{table}[h!]
\caption{\label{tab:gamma-values} In general, as the number of recycled coins $\kappa$ increases, the maximum probability function $G(\kappa, P)$ decreases.}
%\label{tab1}
\begin{center}
%\begin{tabular}{|c|c|c|c|}
%\begin{adjustbox}{width=\textwidth}
\begin{tabular}{ | c | c | c || c | c | c ||c | c | c |}
\hline
$\kappa$ & $P$ & $G(\kappa, P)$ & $\kappa$ & $P$ & $G(\kappa, P)$  & $\kappa$ & $P$ & $G(\kappa, P)$\\ 
%\cline{2-4} 
%\hhline{|=|=|=|=|=|=|=|=|=|=|=|=|}
\hline
2 & 3 & 0.1250  &3 & 3 & 0.0570 &  4 & 3 & 0.0312\\ \hline
2 & 5 & 0.1249  &3 & 5 & 0.0535 &  4 & 5 & 0.0312\\ \hline
2 & 11& 0.0995 &3 & 11& 0.0450 & 4 & 11 & 0.0312\\ \hline
2 & 21& 0.1044 &3 & 21& 0.0282 & 4 & 21 & 0.0272\\ \hline
2 & 51& 0.1057 &3 & 51& 0.0190 & 4 & 51 & 0.0274\\ \hline
\end{tabular}

%\end{adjustbox}
\end{center}
\end{table}
   
\begin{table}[h!]
\caption{\label{tab:gamma-values} In general, for all positional dimensions, as the number of coins increases, the function $G(\kappa, P)$ decreases. But in case of $\kappa =2$, the maximum probability fluctuates.}
\begin{center}
    \begin{tabular}{ | c | c | c || c | c | c || c | c | c |}
    \hline
    $\kappa$ & $P$ & $G(\kappa, P)$ & $\kappa$ & $P$ & $G(\kappa, P)$ & $\kappa$ & $P$ & $G(\kappa, P)$  \\ 
%    \hhline{|=|=|=|=|=|=|=|=|=|}
	\hline
    1 & 3 & 0.1729 &   2 & 3 & 0.1228  &3 & 3 & 0.0614 \\ \hline
    1 & 5 & 0.1133 & 2 & 5 & 0.1251 &3 & 5 & 0.0402 \\ \hline
    1 & 11 & 0.0534 & 2 & 11 & 0.0799 &3 & 11 & 0.0274 \\ \hline
    1 & 21 & 0.0420 & 2 & 21 & 0.0709 &3 & 21 & 0.0192 \\ \hline
    \end{tabular}
    
\end{center}
\end{table}   
%We show the simulations of the random bit rate of the trivial case with the above maximum probabilities $(26)$ and $(27)$ in the following figures. In general, $x$-axis: number of signals sent $N = 1,...,10^{10}$ (from which $m = \sqrt{N}$ are used for sampling); $y$-axis: random bit-generation rate (namely, $\ell_{ours}/N$ where $\ell_{ours}$ is computed using $(17)$.
The random bit rates of the using all case with the Hadamard coin operator show in Fig.1 and Fig.2. They show interesting manners depending on the $\kappa$-number of quantum recycled coins used.
\subsubsection{The odd number of coins improves random bit rates against noises}
The using all case of the memory-based QW-QRNG protocols utilize HD-QW states with the odd number of recycled coins ($\kappa=1, 3$) and the Hadamard coin operator achieves better random bit rates over all noises we tested, e.g., (w(q) = 0, 0.15, 0.2, 0.3). When we increase the number of recycled coins $\kappa$ to add increments, the random bit rates are improved. Particularly if the number of coins $\kappa=3$, the random bit rate over a low dimensional walker's space $|W| = 2^{3}\cdot 3$ shows more resilience as noises are increased than in other walker's spaces, shown in Fig.1. Note that since these simulation results are tested with a limited number of recycled coins, e.g., $\kappa=1,3$, for a clear outcome, it needs to experiment with a higher number of recycled coins such as $\kappa=5,7$, which is restricted as a walker's dimension is increased.
%\begin{figure}[htbp]
\begin{figure}[h!]
\centerline{\includegraphics[width=0.5\textwidth]{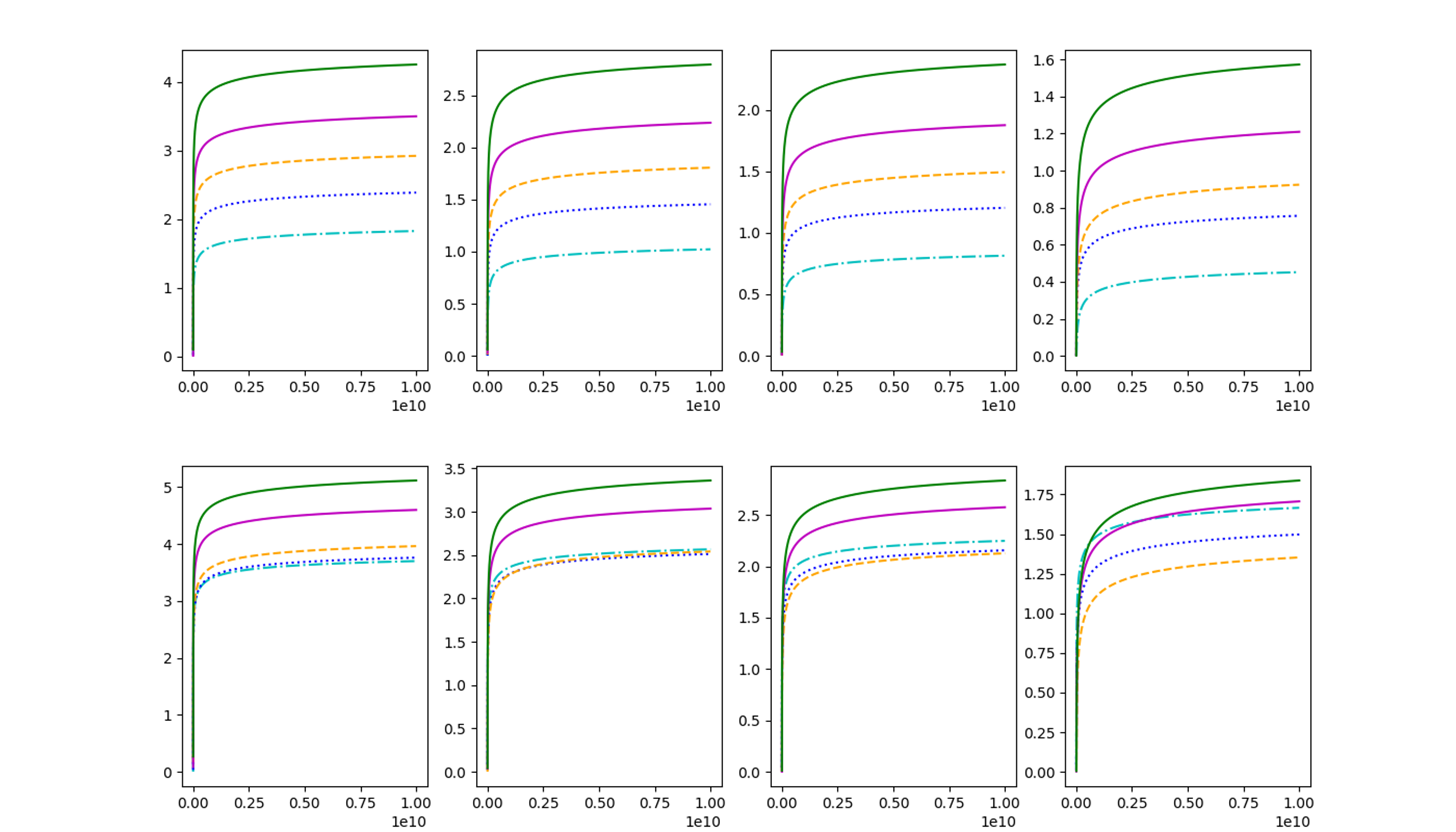}}
\caption{The first row is the case of the number of coins $\kappa = 1$, i.e., the case of the non-history dependent quantum walker; and the second row is the case of the odd-number of coins $\kappa=3$; $x$-axis: number of signals sent $N$; $y$-axis: random bit rate; green solid is $|W| = 2^{\kappa}\cdot 51$;  magenta solid is $|W| = 2^{\kappa}\cdot 21$;  orange-dashed is $|W| = 2^{\kappa}\cdot 11$; blue-dotted is $|W| = 2^{\kappa}\cdot 5$; cyan-dot-dashed is $|W| = 2^{\kappa}\cdot 3$; left graph is without noise ($w(q)=0$); second left graph is with $15\%$ noise in the source; second right graph has $20\%$ noise; right graph has $30\%$ noise.}
\label{fig:trivial_case_01}
\end{figure}

\subsubsection{The even number of coins enhances random bit rates of the low-dimensional walker's spaces against noises}
The memory-based QW-QRNG protocols with HD-QW states of the even number of recycled coins ($\kappa=2, 4$) and the Hadamard coin operator show the random bit rates over low dimensional walker's spaces, e.g., $|W| = 2^{2}\cdot 3, |W|=2^{4}\cdot 3$ are able to withstand to noises are increased than in high walker's spaces, shown in Fig.2. In other words, the random bit rate over a high dimensional walker's space, e.g., $|W|=2^{2}\cdot 51$, is vulnerable as noises are increased. But as the number of recycled coins $\kappa$ to even increments, the noise sensitivity is improved while the random bit rate over lower dimensions still shows better resilience than higher ones, shown in Fig.2. Similarly, since these simulation results are tested with a limited number of recycled coins, e.g., $\kappa=2,4$, for a better result, it needs to test with a higher number of recycled coins such as $\kappa=6,8$, which is restrained as a walker's dimension is increased.
\begin{figure}[h!]
    \centering
    \includegraphics[width=0.5\textwidth]{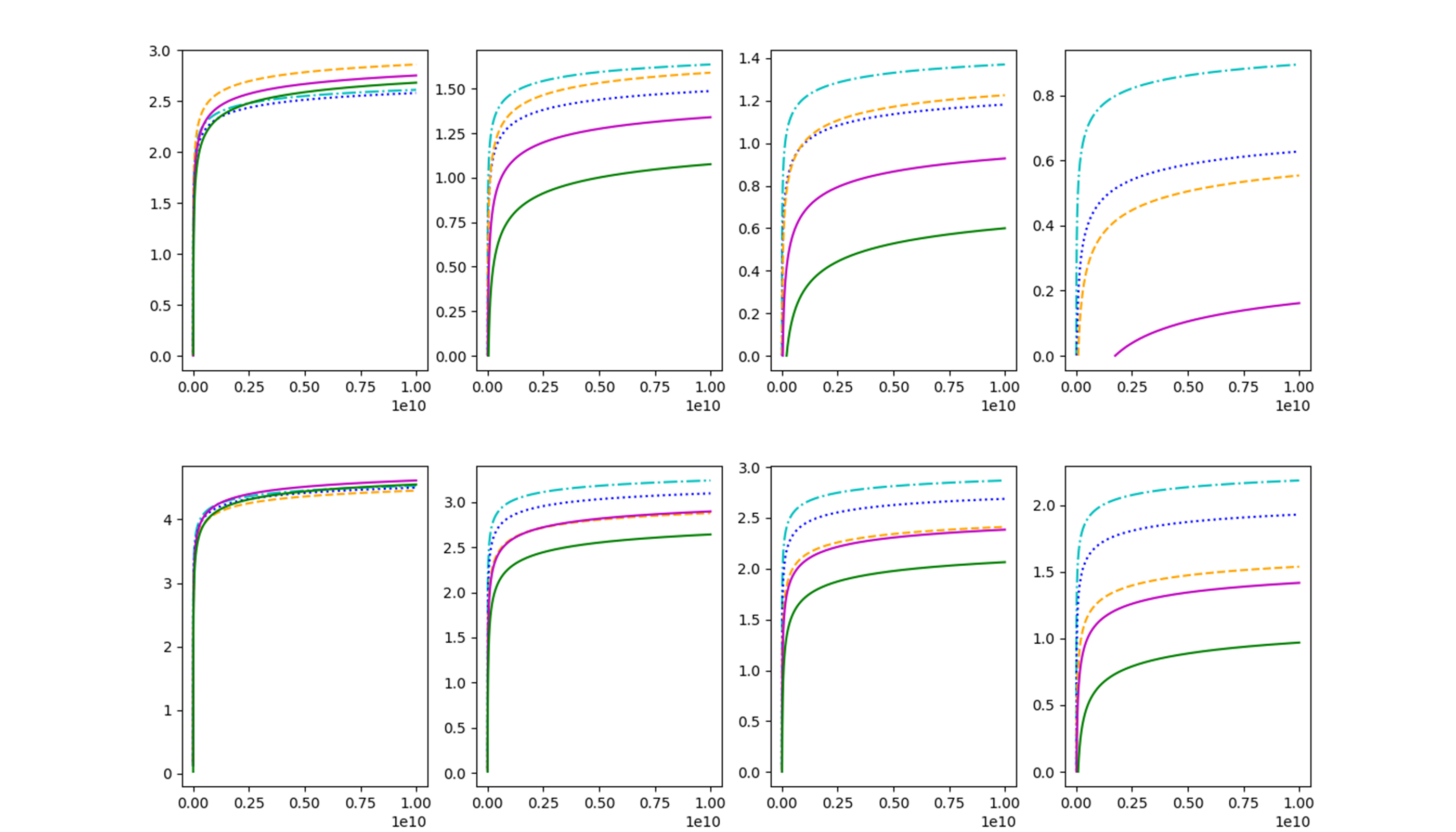}
    \caption{The first row is the case of the even-number of coins $\kappa = 2$; and the second row is the case of the another even-number of coins $\kappa=4$; $x$-axis is the number of signals sent $N$; $y$-axis is the random bit rate; green solid is $|W| = 2^{\kappa}\cdot 51$;  magenta solid is $|W| = 2^{\kappa}\cdot 21$; orange-dashed is $|W| = 2^{\kappa}\cdot 11$; blue-dotted is $|W| = 2^{\kappa}\cdot 5$; cyan-dot-dashed is $|W| = 2^{\kappa}\cdot 3$; left graph is without noise ($w(q)=0$); second left graph is with $15\%$ noise in the source; second right graph has $20\%$ noise; right graph has $30\%$ noise.}
    \label{fig:trivial_case_02}
\end{figure}

%As you see Fig. 1 and 2, the case with the low-positional dimension $P=5$ shows that secure random bit rate is more noise-tolerant than the cases of the higher positional space of the walker.

\subsubsection{The generalized and flip coin operators raise random bit rates}
First, the memory-based QW-QRNG using the generalized and flip-coin operators (\ref{eq:general-coin-op}) and (\ref{eq:flip-coin-ops}) improves its overall random bit rates, as shown in Fig.3. Secondly, the noise sensitivity in random bit rates over the higher walker's dimensions is improved, particularly in the cases of the even number of recycled coins ($\kappa=2$), shown in Fig.3.
\begin{figure}[h!]
    \centering
    \includegraphics[width=0.5\textwidth]{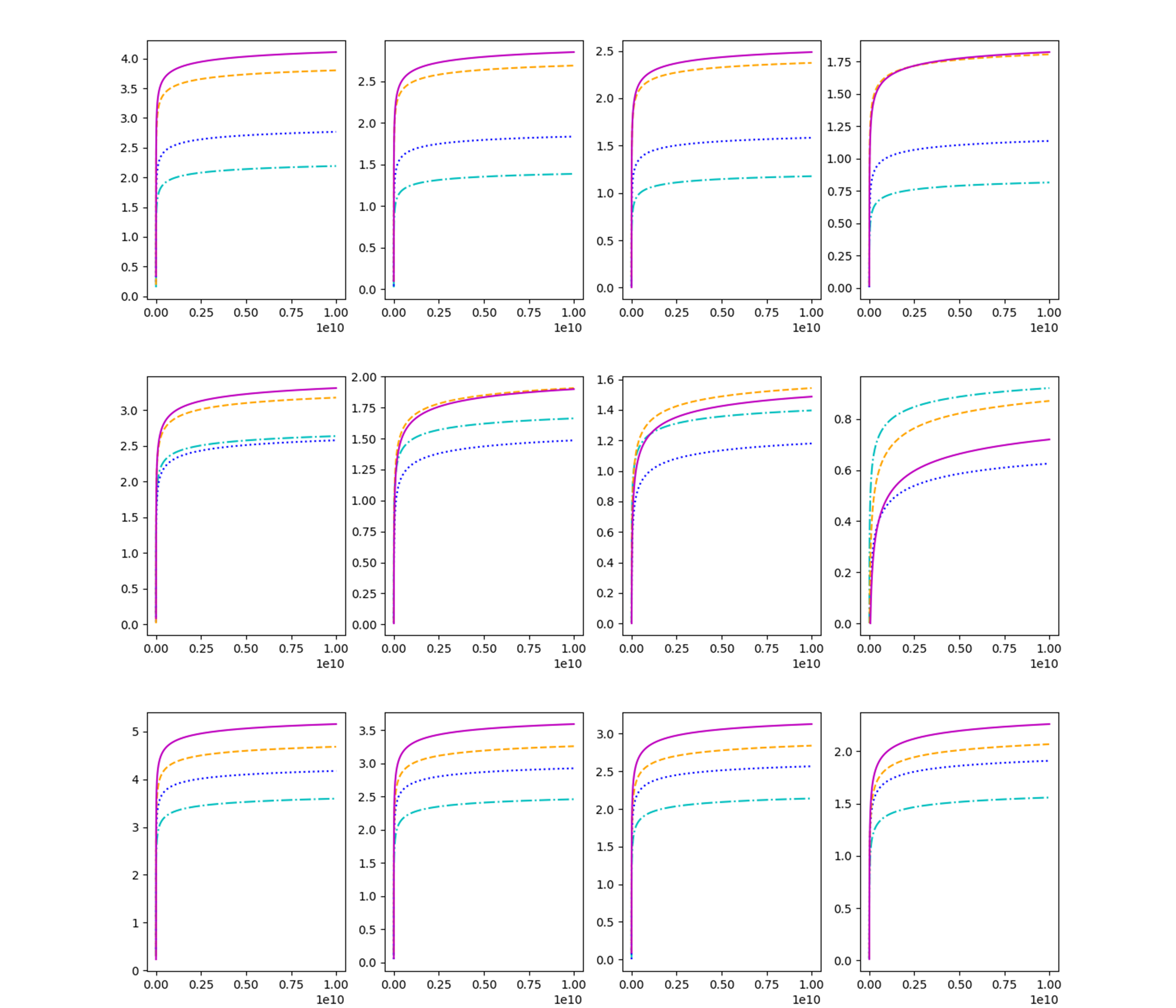}
    \caption{The first row is the case of the number of coins $\kappa = 1$; the second row is the case of the number of coins $\kappa=2$; and the last row is the case of the number of coins $\kappa=3$; $x$-axis is the number of signals sent $N$; $y$-axis is the random bit rate; magenta solid is $|W| = 2^{\kappa}\cdot 21$;  orange-dashed is $|W| = 2^{\kappa}\cdot 11$; blue-dotted is $|W| = 2^{\kappa}\cdot 5$; cyan-dot-dashed is $|W| = 2^{\kappa}\cdot 3$; left graph is without noise ($w(q)=0$); second left graph is with $15\%$ noise in the source; second right graph has $20\%$ noise; right graph has $30\%$ noise.}
    \label{fig:optimized_trivial_case_01}
\end{figure}

%\newpage

\subsection{The Using Memory Case}

%To evaluate the using memory case's random bit rate, the gamma $\gamma'$ is defined as follows: given an honest walker's state $\ket{w_{0,0,...,0}}$, for any random digit $z\in\{0,...,P-1\}$ and random memory coins $c_{\mu} = c_{0},...,c_{\kappa-2}\in\{0,1\}^{\mu}$,
%\begin{align}
%\gamma' = -\log_{2}\max_{z,c_{\mu}} \mathbb{P}_{W}\big(\ket{w_{0,0,...,0}}\to (z,c_{\mu})\big).
%\end{align} 

In the using memory case, the protocol exploits only the $\mu$-number of memory coins to produce a string of true random bits. The secure random bits rates are to compute via Eq. (\ref{eq:gamma-01}) in Theorem 3.  For $\kappa = 1$, which is the non-memory based quantum walk case, the case has the following maximum guessing probabilities (\ref{eq:max-prob-01}): $G'(1,3) = 0.3634$, $G'(1,5) = 0.2447$, $G'(1,11) = 0.1358$, $G'(1,21) = 0.0919$, and $G'(1,51) = 0.0517$. The using memory case with the maximum probability (\ref{eq:max-prob-01}) and the maximum probability (\ref{eq:general-max-prob-01}) are computed in table III and IV, respectively. Note that in the table IV, the function $G'(\kappa, P, F)$ is written as $G'(\kappa, P)$.
\begin{table}[h!]
\caption{\label{tab:gamma-values} In general, as the number of coins $\kappa$ increases, the function $G'(\kappa, P)$ decreases. But, in case of $\kappa = 2$, the maximum probability fluctuates.}
\begin{center}
\begin{tabular}{ | c | c | c || c | c | c ||c | c | c |}
\hline
$\kappa$ & $P$ & $G'(\kappa, P)$ & $\kappa$ & $P$ & $G'(\kappa, P)$  & $\kappa$ & $P$ & $G'(\kappa, P)$\\ 
%\cline{2-4} 
%\hhline{|=|=|=|=|=|=|=|=|=|=|=|=|}
\hline
2 & 3 & 0.2500  &3 & 3 & 0.1120 &  4 & 3 & 0.0625\\ \hline
2 & 5 & 0.1875  &3 & 5 & 0.0656 &  4 & 5 & 0.0617\\ \hline
2 & 11& 0.1378 &3 & 11& 0.0524 & 4 & 11 & 0.0453\\ \hline
2 & 21& 0.1342 &3 & 21& 0.0374 & 4 & 21 & 0.0340\\ \hline
2 & 51& 0.1377 &3 & 51& 0.0233 & 4 & 51 & 0.0314\\ \hline
\end{tabular}
    
    \end{center}
   \end{table}
\begin{table}[h!]
\caption{\label{tab:gamma-values} In general, for all positional dimensions, as the number of coins increases, the function $G'(\kappa, P)$ decreases. But in case of $\kappa =2$, the maximum probability fluctuates.}
\begin{center}
    \begin{tabular}{ | c | c | c || c | c | c || c | c | c |}
    \hline
    $\kappa$ & $P$ & $G'(\kappa, P)$ & $\kappa$ & $P$ & $G'(\kappa, P)$ & $\kappa$ & $P$ & $G'(\kappa, P)$  \\ 
    \hline
    1 & 3 & 0.3334&   2 & 3 & 0.1751  &3 & 3 & 0.0898 \\ \hline
    1 & 5 & 0.2017 & 2 & 5 & 0.1615 &3 & 5 & 0.0661 \\ \hline
    1 & 11 & 0.0952 & 2 & 11 & 0.1082 &3 & 11 & 0.0417 \\ \hline
    1 & 21 & 0.0617 & 2 & 21 & 0.0743 &3 & 21 & 0.0264 \\ \hline
    \end{tabular}
    
\end{center}
\end{table}

The random bit rates of the using memory case with the Hadamard coin operator and the generalized coin operator show in Fig.4 and Fig.5, respectively.
\subsubsection{Increasing the number of memory coins improves random bit rates against overall noises}
As the size of the memory coin space is increased, we see that the using memory case is resilient to noises so that it shows better random bit rates than one that has a small memory coin space ($\kappa=2$), see Fig.4. Note that the case of $\kappa=1$ has no memory coin space in an HD-QW state. In the case of $\kappa=2$, there is only one memory coin $\mu= \kappa-1$, the random rate shows that the protocol is vulnerable to noises, e.g., $w(q)=0.2$. But as the number of memory coins is increased, the random bit rates in Fig.4 show that the using memory case recovers quickly from noises. 
%Similar to the using all case, for all positional dimensions, the case of the odd number of recycled coins shows better the random bit rates than the non-history walker's protocol ($\kappa=1$) overall noises in Fig.4. Because this protocol does not use the active coin for the random number generation, the memory-only protocol looks more vulnerable to higher noises ($w(q)=0.2$), than the trivial case.
\begin{figure}[h!]
    \centering
    \includegraphics[width=0.5\textwidth]{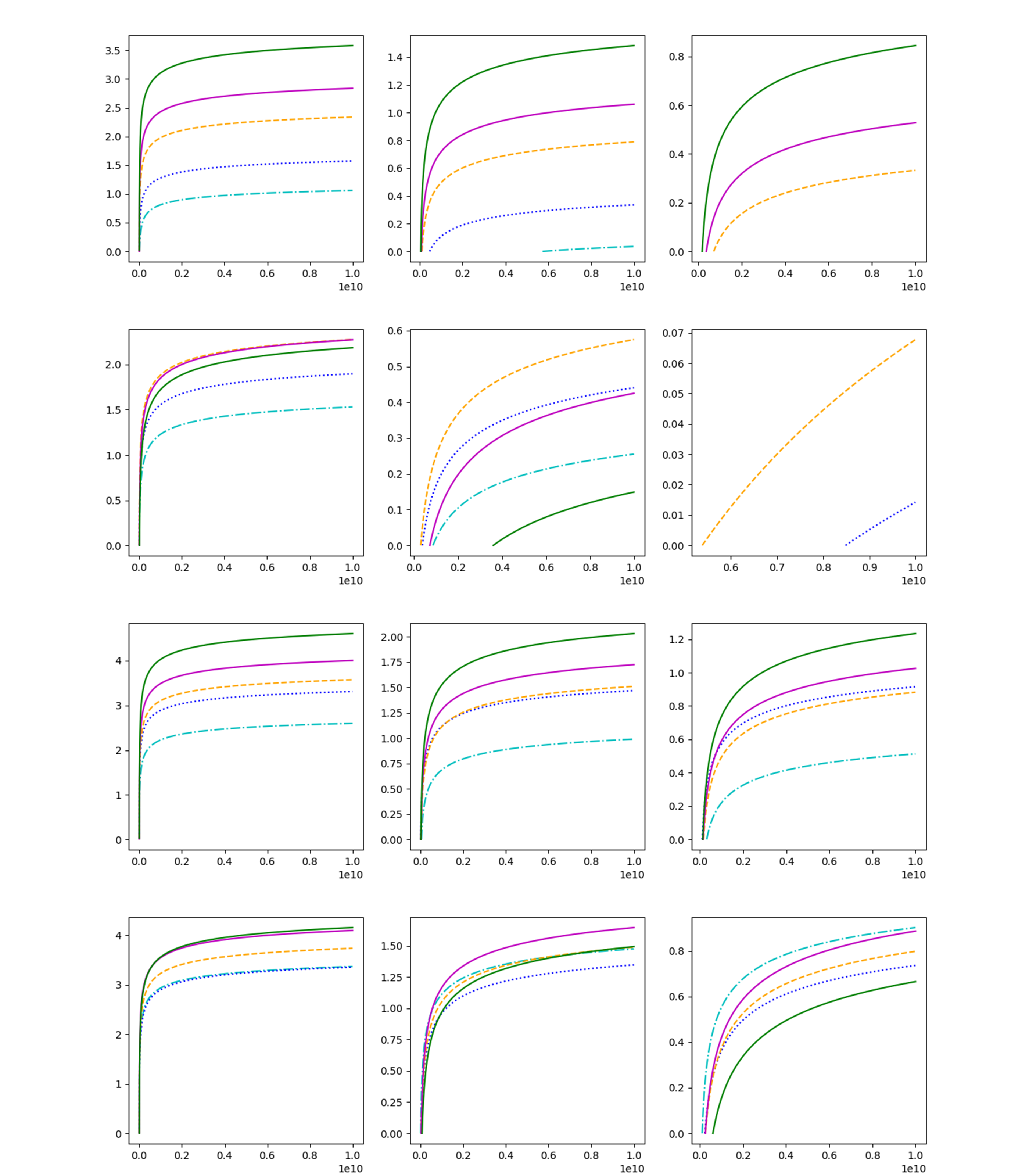}
    \caption{The first row is the case of the number of coins $\kappa = 1$, that is, the case of the non-history dependent quantum walker; the second row is the case of $\kappa=2$; the third row is the case of $\kappa=3$; and the last row is the case of $\kappa=4$; $x$-axis is the number of signals sent $N$; $y$-axis is the random bit rate; green solid is $|W| = 2^{\kappa}\cdot 51$; magenta solid is $|W| = 2^{\kappa}\cdot 21$;  orange-dashed is $|W| = 2^{\kappa}\cdot 11$; blue-dotted is $|W| = 2^{\kappa}\cdot 5$; cyan-dot-dashed is $|W| = 2^{\kappa}\cdot 3$; left graph is without noise ($w(q)=0$); middle graph is with $15\%$ noise in the source; right graph has $20\%$ noise.}
    \label{fig:using_memory_case_01}
\end{figure}

%\subsubsection{The even number of coins enhances random bit rates of the low-dimensional positions}
%As the even number of coins increases, the lower positional dimension $P=3$ is more resilient to noises than the higher positional dimensions. But in the case of $\kappa =2$, the protocol is unsafe to higher noise ($w(q) = 0.2$) for overall positions. The random bit rates with the even-number of coins shows in Fig. 5.
%\begin{figure}[h!]
%    \centering
%    \includegraphics[width=0.5\textwidth]{Figure_04}
%    \caption{The first row is the case of the even-number of coins $\kappa = 2$; and the second row is the case of the another even-number of coins $\kappa=4$; $x$-axis is the number of signals sent $N$; $y$-axis is the random bit rate; green solid is $|W| = 2^{2}\cdot 51$; magenta solid is $|W| = 2^{2}\cdot 21$;  orange-dashed is $|W| = 2^{2}\cdot 11$; blue-dotted is $|W| = 2^{2}\cdot 5$; cyan-dot-dashed is $|W| = 2^{2}\cdot 3$; left graph is without noise ($w(q)=0$); middle graph is with $15\%$ noise in the source; right graph has $20\%$ noise.}
%    \label{fig:using_memory_case_02}
%\end{figure}

\subsubsection{The generalized and flip coin operators increase the overall random bit rates}
Using the generalized and flip-coin operators helps improve the using memory case's random bit rates. The protocol's vulnerability having the small memory coin space to noises is improved, see Fig.5.

\begin{figure}[h!]
    \centering
    \includegraphics[width=0.5\textwidth]{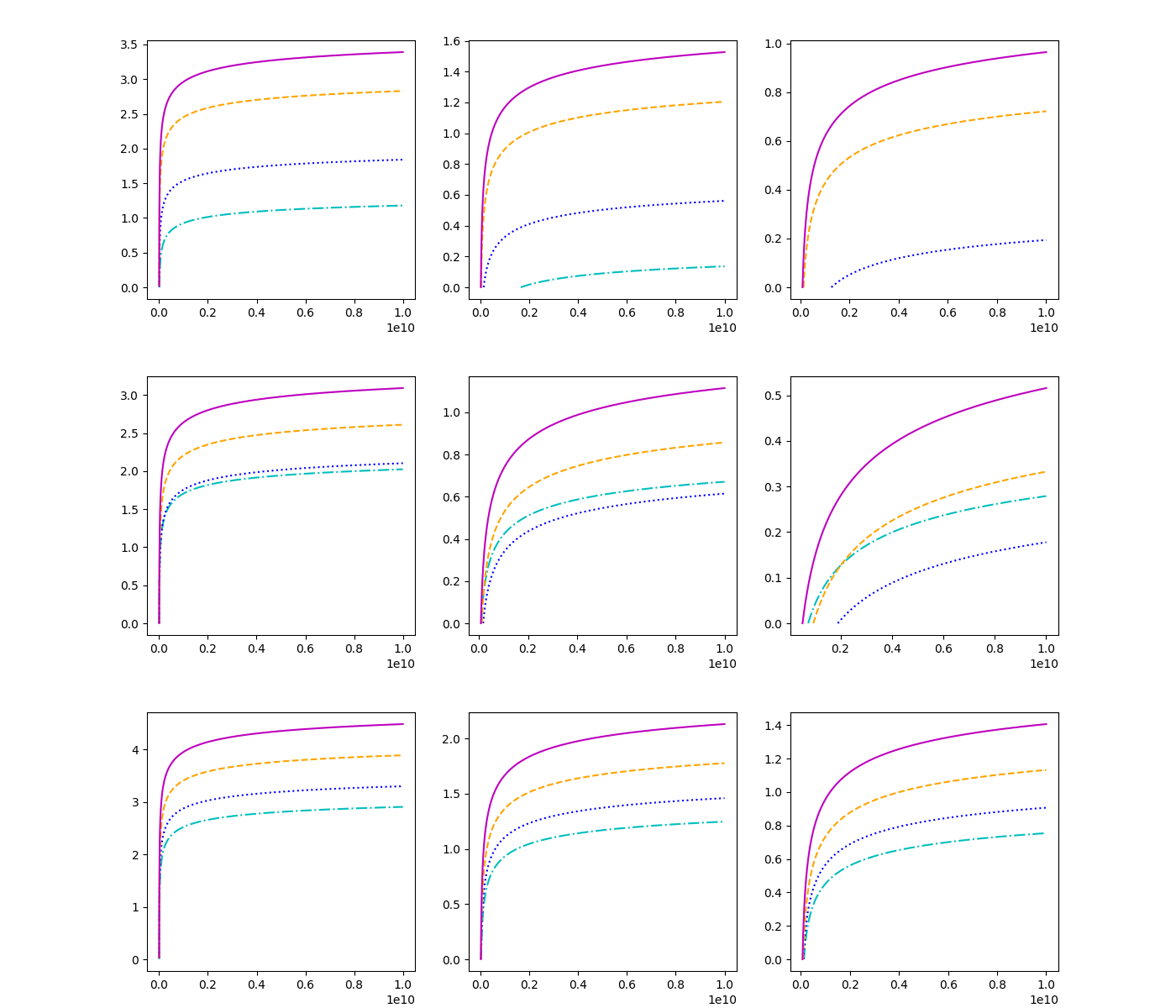}
    \caption{The first row is the case of the number of coins $\kappa = 1$; the second row is the case of the number of coins $\kappa=2$; and the last row is the case of the number of coins $\kappa=3$; $x$-axis is the number of signals sent $N$; $y$-axis is the random bit rate; magenta solid is $|W| = 2^{\kappa}\cdot 21$;  orange-dashed is $|W| = 2^{\kappa}\cdot 11$; blue-dotted is $|W| = 2^{\kappa}\cdot 5$; cyan-dot-dashed is $|W| = 2^{\kappa}\cdot 3$; left graph is without noise ($w(q)=0$); middle graph is with $15\%$ noise in the source; right graph has $20\%$ noise.}
    \label{fig:optimized_using_memory_case_01}
\end{figure}

\subsection{The Not Using Memory Case}
In the not using memory case, the protocol does not use any recycled (memory and active) coins to generate a string of true random bits. The secure random bits rates are computed by using Eq. (\ref{eq:gamma-02}) in Theorem 3. For $\kappa = 1$, which is the non-memory based quantum walk case, the case has the following maximum guessing probabilities (\ref{eq:general-max-prob-02}): $G''(1,3) = 0.3634$, $G''(1,5) = 0.2447$, $G''(1,11) = 0.1358$, $G''(1,21) = 0.0919$, and $G''(1,51) = 0.0517$. For $\kappa >1$, the not using memory case with the maximum probability (\ref{eq:max-prob-02}) and the maximum probability (\ref{eq:general-max-prob-02}) are evaluated in table V and VI, respectively. Note that in the table VI, the function $G''(\kappa, P, F)$ is written as $G'(\kappa, P)$. For all positions, the maximum probability $G''(\kappa, P)$ of the number of coins $\kappa =4$ will be computed to have better insight. When the protocol excludes to use all memory and active coins to generate random bits, the protocol becomes very vulnerable to the noises. 

\begin{table}[h!]
\caption{\label{tab:gamma-values} As the number of coins $\kappa$ increases, the function $G''(\kappa, P)$ does not decrease much. Especially, for higher position dimensions, e.g., $P=21, 51$, the function $G''(\kappa, P)$ goes up and down.}
\begin{center}
    \begin{tabular}{ | c | c | c || c | c | c ||c | c | c |}
    \hline
   $\kappa$ & $P$ & $G''(\kappa, P)$ & $\kappa$ & $P$ & $G''(\kappa, P)$  & $\kappa$ & $P$ & $G''(\kappa, P)$\\ 
    \hline
    2 & 3 & 0.3336  &3 & 3 & 0.3400 &  4 & 3 & 0.3437\\ \hline
    2 & 5 & 0.2570  &3 & 5 & 0.2165 &  4 & 5 & 0.2055\\ \hline
    2 & 11& 0.1831 &3 & 11& 0.1186 & 4 & 11 & 0.1230\\ \hline
    2 & 21& 0.1692 &3 & 21& 0.0778 & 4 & 21 & 0.0808\\ \hline
    2 & 51& 0.1701 &3 & 51& 0.0379 & 4 & 51 & 0.0709\\ \hline
    \end{tabular}
    
    \end{center}
   \end{table}
\begin{table}[h!]
\caption{\label{tab:gamma-values} In general, as the number of coins $\kappa$ increases, the function $G''(\kappa, P)$ does not lessen much.}
\begin{center}
    \begin{tabular}{ | c | c | c || c | c | c || c | c | c |}
    \hline
    $\kappa$ & $P$ & $G''(\kappa, P)$ & $\kappa$ & $P$ & $G''(\kappa, P)$ & $\kappa$ & $P$ & $G''(\kappa, P)$  \\ 
    \hline
    1 & 3 & 0.3334 &   2 & 3 & 0.3340  &3 & 3 & 0.3336 \\ \hline
    1 & 5 & 0.2017 & 2 & 5 & 0.2197 &3 & 5 & 0.2097 \\ \hline
    1 & 11 & 0.0952 & 2 & 11 & 0.1275 &3 & 11 & 0.1039 \\ \hline
    1 & 21 & 0.0617 & 2 & 21 & 0.0834 &3 & 21 & 0.0642 \\ \hline
    \end{tabular}
    
    \end{center}
   \end{table}

\subsubsection{Increasing the number of coins does not help remedy the vulnerability to noise}
The not using memory protocol is vulnerable to noises. Also, increasing the number of coins in the walker's evolution does not help remedy noise vulnerability and improve random bit rates, see Fig.6.
\begin{figure}[h!]
    \centering
    \includegraphics[width=0.5\textwidth]{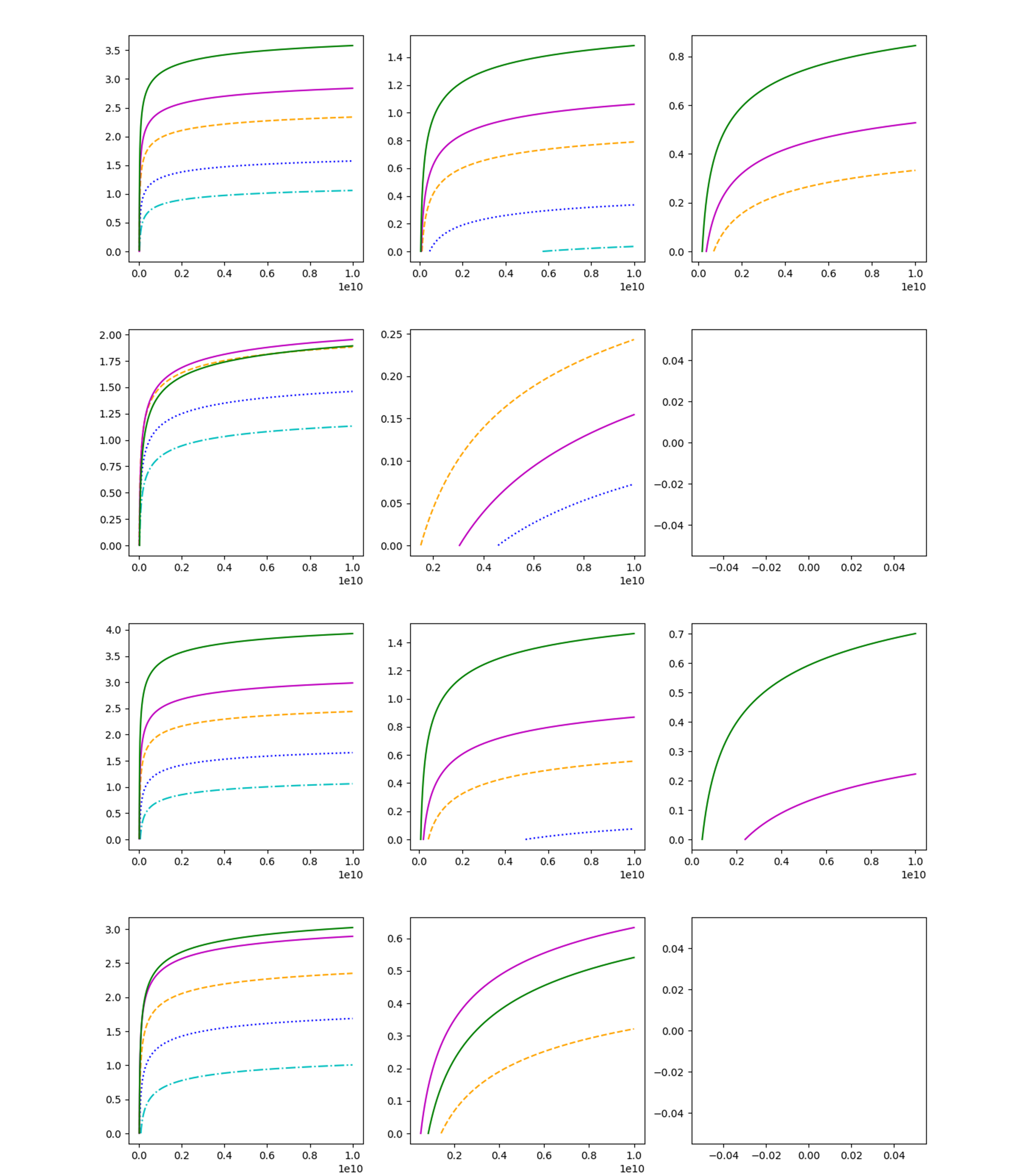}
    \caption{The first row is the case of the number of coins $\kappa = 1$; the second row is the case of the number of coins $\kappa=2$; the third row is the case of the number of coins $\kappa=3$; the last row is the case of the number of coins $\kappa=4$; $x$-axis is the number of signals sent $N$; $y$-axis is the random bit rate; green solid is $|W| = 2^{\kappa}\cdot 51$; magenta solid is $|W| = 2^{\kappa}\cdot 21$;  orange-dashed is $|W| = 2^{\kappa}\cdot 11$; blue-dotted is $|W| = 2^{\kappa}\cdot 5$; cyan-dot-dashed is $|W| = 2^{\kappa}\cdot 3$; left graph is without noise ($w(q)=0$); middle graph is with $15\%$ noise in the source; right graph has $20\%$ noise.}
    \label{fig:not_using_memory_case_01}
\end{figure}

\subsubsection{The generalized and flip coin operators do not help to improve the vulnerability to noise}
Using the generalized and flip-coin operators and increasing the number of coins do not help remedy noise vulnerability and improve random bit rates. It looks when the protocol is not using memory coins, and the protocol is vulnerable to noises overall the positional dimensions as the number of coins increases, see Fig.7.
\begin{figure}[h!]
    \centering
    \includegraphics[width=0.5\textwidth]{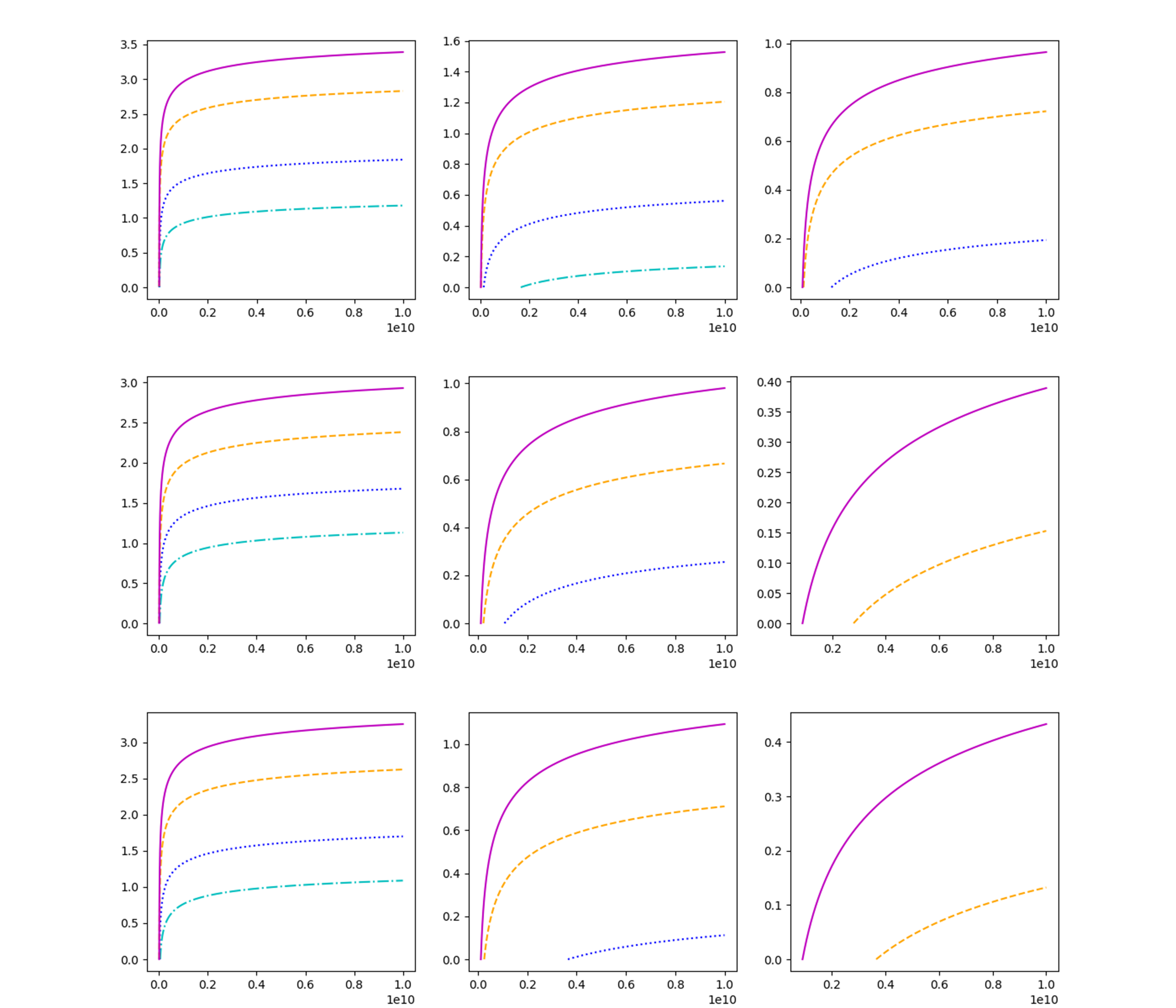}
    \caption{The first row is the case of the number of coins $\kappa = 1$; the second row is the case of the number of coins $\kappa=2$; the last row is the case of the number of coins $\kappa=3$; $x$-axis is the number of signals sent $N$; $y$-axis is the random bit rate; magenta solid is $|W| = 2^{\kappa}\cdot 21$;  orange-dashed is $|W| = 2^{\kappa}\cdot 11$; blue-dotted is $|W| = 2^{\kappa}\cdot 5$; cyan-dot-dashed is $|W| = 2^{\kappa}\cdot 3$. left graph is without noise ($w(q)=0$); middle graph is with $15\%$ noise in the source; right graph has $20\%$ noise.}
    \label{fig:optimized_not_using_memory_case_02}
\end{figure}

\section{Closing Remarks}
In this paper, we newly devise various memory-based QW-QRNG protocols by using an HD-QW state \cite{rohde2013quantum}. We explore the memory-based QW-QRNG protocols in multiple scenarios, including the using all case (using memory and active coins), the using memory coins, and the not using memory coins. We analyze these memory-based QW-QRNG protocols to be secure in the semi-source independent (SI) model. We simulate the protocols with different coin operators (Hadamard, generalized, and flip-coin operators) to optimize the randomness of the bit string. Throughout the simulations, we show exciting behaviors of the protocols depending on the size of the memory space and the number of quantum recycled coins. For example, in the protocol using the odd number of recycled coins to generate a random bit, the protocol can improve the random bit rate against overall noises. But when the protocol uses the even number of recycled coins, it enhances the random bit rate of the low-dimensional position space against general noises. This simulation result may connect to some open problems, particularly analyzing how the number of recycled coins affects the random bit rate of the memory-based QW-QRNG over noises. Also, developing and analyzing a protocol of QKD with the memory-based quantum walk state will be another exciting problem. Moreover, devising a non-local game \cite{clauser1969chsh, francesco2012nonlocal, matthew2017nonlocal, ming2019nonlocal, johnny2022nonlocal} with an entangled state via quantum walks \cite{Ivens2005entangled, omar2006entangled, abal2006entangled, salvador2009entangled, yusuke2010entangled, scott2011entangled, semra2012entangled, meng2021entangled} will be an exciting project so that we may exploit it to develop a device-independent QRNG/QKD protocol \cite{artur1992qc_bell, stefano2009diqkd, lluis2011diqkd, charles2013diqkd, yang2018diqrng, umesh2019diqkd, yanbao2020diqrng, david2021diqkd} without the measurement independent assumption in the non-local game \cite{thomas2010freedom_choice, bernhard2012freedom_choice, andrew2019freedom_choice, michael2020freedom_choice, rafael2021freedom_choice}.

\section*{Acknowledgment}

The author would like to thank Walter O. Krawec for valuable feedback and helpful discussions. Also, the author relishes significant comments from the anonymous critics, which seriously ripens the paper's quality.

%\newpage


\begin{thebibliography}{00}
%\begin{thebibliography}{9}

%%%%%%%%%%%%%%%%%%%%%%%%%%%%%%%%%%%%%%%%%%%%%%%%%%%%%%%%%%%%%%%%%%%%
% qrng
\bibitem{di-qrng1}
Colbeck, Roger, and Adrian Kent. ``Private randomness expansion with untrusted devices." Journal of Physics A: Mathematical and Theoretical 44.9 (2011): 095305.
  
\bibitem{di-qrng0} 
Vazirani, Umesh, and Thomas Vidick. ``Certifiable quantum dice: or, true random number generation secure against quantum adversaries." Proceedings of the forty-fourth annual ACM symposium on Theory of computing. 2012.

\bibitem{di-qrng2}
Pironio, Stefano, and Serge Massar. ``Security of practical private randomness generation." Physical Review A 87.1 (2013): 012336.

\bibitem{frauchiger2013true}
Frauchiger, Daniela, Renato Renner, and Matthias Troyer. ``True randomness from realistic quantum devices." arXiv preprint arXiv:1311.4547 (2013).
  
\bibitem{vallone2014quantum}
Vallone, Giuseppe, et al. ``Quantum randomness certified by the uncertainty principle." Physical Review A 90.5 (2014): 052327.
  
\bibitem{xu2016experimental}
Xu, Feihu, Jeffrey H. Shapiro, and Franco NC Wong. ``Experimental fast quantum random number generation using high-dimensional entanglement with entropy monitoring." Optica 3.11 (2016): 1266-1269.
  
\bibitem{di-qrng-exp}
Bierhorst, Peter, et al. ``Experimentally generated randomness certified by the impossibility of superluminal signals." Nature 556.7700 (2018): 223-226.

\bibitem{di-qrng-exp-2}
Liu, Yang, et al. ``High-speed device-independent quantum random number generation without a detection loophole." Physical review letters 120.1 (2018): 010503.
    
\bibitem{si-qrng-fast}
Avesani, M., et al. ``Secure heterodyne-based quantum random number generator at 17 Gbps (2018)." arXiv preprint arXiv:1801.04139.

\bibitem{si-qrng-sun}
Li, Yu-Huai, et al. ``Quantum random number generation with uncharacterized laser and sunlight." npj Quantum Information 5.1 (2019): 1-5.

\bibitem{qrng-survey}
Herrero-Collantes, Miguel, and Juan Carlos Garcia-Escartin. ``Quantum random number generators." Reviews of Modern Physics 89.1 (2017): 015004.
  
%%%%%%%%%%%%%%%%%%%%%%%%%%%%%%%%%%%%%%%%%%%%%%%%%%%%%%%%%%%%%%%%%%%%
% quantum algorithms
\bibitem{mosca2008qalgorithms}
Mosca, Michele. ``Quantum algorithms." arXiv preprint arXiv:0808.0369 (2008).

\bibitem{montanaro2016qalgorithms}
Montanaro, Ashley. ``Quantum algorithms: an overview." npj Quantum Information 2.1 (2016): 1-8.
\bibitem{cerezo2021qalgorithms}
Cerezo, Marco, et al. ``Variational quantum algorithms." Nature Reviews Physics 3.9 (2021): 625-644.

\bibitem{bharti2022qalgorithms}
Bharti, Kishor, et al. ``Noisy intermediate-scale quantum algorithms." Reviews of Modern Physics 94.1 (2022): 015004.

  
%%%%%%%%%%%%%%%%%%%%%%%%%%%%%%%%%%%%%%%%%%%%%%%%%%%%%%%%%%%%%%%%%%%%
% quantum walk

\bibitem{farhi1998quantum}
E.~Farhi and S.~Gutmann, ``Quantum computation and decision trees,'' Physical Review A, vol.~58, no.~2, p. 915, 1998.
  
\bibitem{aharonov2001quantum}
Aharonov, Dorit, et al. ``Quantum walks on graphs." Proceedings of the thirty-third annual ACM symposium on Theory of computing. 2001.
  
\bibitem{bednarska2003qwalk}
Bednarska, Małgorzata, et al. ``Quantum walks on cycles." Physics Letters A 317.1-2 (2003): 21-25.

\bibitem{childs2003exponential}
Childs, Andrew M., et al. ``Exponential algorithmic speedup by a quantum walk." Proceedings of the thirty-fifth annual ACM symposium on Theory of computing. 2003.
   
\bibitem{childs2009universal}
Childs, Andrew M. ``Universal computation by quantum walk." Physical review letters 102.18 (2009): 180501.

\bibitem{lovett2010universal}
Lovett, Neil B., et al. ``Universal quantum computation using the discrete-time quantum walk." Physical Review A 81.4 (2010): 042330.

\bibitem{renato2013qwalk}  
Portugal, Renato. ``Quantum walks and search algorithms". Vol. 19. New York: Springer, 2013.

\bibitem{proctor2014non} 
Proctor, T. J., et al. ``Nonreversal and nonrepeating quantum walks." Physical Review A 89.4 (2014): 042332.

\bibitem{ashley2016qalgorithms} 
Montanaro, Ashley. ``Quantum algorithms: an overview." npj Quantum Information 2.1 (2016): 1-8.

\bibitem{santha2008qwalgorithms} 
Santha, Miklos. ``Quantum walk based search algorithms." International Conference on Theory and Applications of Models of Computation. Springer, Berlin, Heidelberg, 2008.

\bibitem{magniez2011qwalgorithms} 
Magniez, Frédéric, et al. ``Search via quantum walk." SIAM journal on computing 40.1 (2011): 142-164.

\bibitem{balu2017qwalgorithms} 
Balu, Radhakrishnan, Chaobin Liu, and Salvador E. Venegas-Andraca. ``Probability distributions for Markov chain based quantum walks." Journal of Physics A: Mathematical and Theoretical 51.3 (2017): 035301.

\bibitem{kadian2021quantumwalk} 
Kadian, Karuna, Sunita Garhwal, and Ajay Kumar. ``Quantum walk and its application domains: A systematic review." Computer Science Review 41 (2021): 100419.

%\bibitem{bezerra2021quantumwalk} 
%Bezerra, G. A., Pedro HG Lugão, and Renato Portugal. "Quantum walk-based search algorithms with multiple marked vertices." arXiv preprint arXiv:2103.12878 (2021).

%%%%%%%%%%%%%%%%%%%%%%%%%%%%%%%%%%%%%%%%%%%%%%%%%%%%%%%%%%%%%%%%%%%%
% memory-based quantum walk
\bibitem{brun2003quantum}
Brun, Todd A., Hilary A. Carteret, and Andris Ambainis. ``Quantum walks driven by many coins." Physical Review A 67.5 (2003): 052317.

\bibitem{faj_2004} 
Flitney, Adrian P., Derek Abbott, and Neil F. Johnson. ``Quantum walks with history dependence." Journal of Physics A: Mathematical and General 37.30 (2004): 7581.

\bibitem{mcgettrick2009one}
McGettrick, Michael. ``One dimensional quantum walks with memory." arXiv preprint arXiv:0911.1653 (2009).

\bibitem{mcgettrick2014cycle} 
Mc Gettrick, Michael, and Jarosław Adam Miszczak. ``Quantum walks with memory on cycles." Physica A: Statistical Mechanics and its Applications 399 (2014): 163-170.

\bibitem{rohde2013quantum}
Rohde, Peter P., Gavin K. Brennen, and Alexei Gilchrist. ``Quantum walks with memory provided by recycled coins and a memory of the coin-flip history." Physical Review A 87.5 (2013): 052302.

\bibitem{krawec2015history}
Krawec, Walter O. ``History dependent quantum walk on the cycle with an unbalanced coin." Physica A: Statistical Mechanics and its Applications 428 (2015): 319-331.

%%%%%%%%%%%%%%%%%%%%%%%%%%%%%%%%%%%%%%%%%%%%%%%%%%%%%%%%%%%%%%%%%%%%
% entangled quanutm walk states

\bibitem{Ivens2005entangled}
Carneiro, Ivens, et al. ``Entanglement in coined quantum walks on regular graphs." New Journal of Physics 7.1 (2005): 156.

\bibitem{omar2006entangled}
Omar, Y., et al. ``Quantum walk on a line with two entangled particles." Physical Review A 74.4 (2006): 042304.

\bibitem{abal2006entangled}
Abal, G., et al. ``Quantum walk on the line: Entanglement and nonlocal initial conditions." Physical Review A 73.4 (2006): 042302.

\bibitem{salvador2009entangled}
Venegas-Andraca, Salvador E., and Sougato Bose. ``Quantum walk-based generation of entanglement between two walkers." arXiv preprint arXiv:0901.3946 (2009).

\bibitem{yusuke2010entangled}
Ide, Yusuke, Norio Konno, and Takuya Machida. ``Entanglement for discrete-time quantum walks on the line." arXiv preprint arXiv:1012.4164 (2010).

\bibitem{scott2011entangled}
Berry, Scott D., and Jingbo B. Wang. ``Two-particle quantum walks: Entanglement and graph isomorphism testing." Physical Review A 83.4 (2011): 042317.

\bibitem{semra2012entangled}
Allés, B., Semra Gündüç, and Yigit Gündüç. ``Maximal entanglement from quantum random walks." Quantum Information Processing 11.1 (2012): 211-227.

\bibitem{meng2021entangled}
Li, Meng, and Yun Shang. ``Entangled state generation via quantum walks with multiple coins." npj Quantum Information 7.1 (2021): 1-8.

%%%%%%%%%%%%%%%%%%%%%%%%%%%%%%%%%%%%%%%%%%%%%%%%%%%%%%%%%%%%%%%%%%%%
% quantum walk qrng/qkd

\bibitem{rohde2012quantum}
Rohde, Peter P., Joseph F. Fitzsimons, and Alexei Gilchrist. ``Quantum walks with encrypted data." Physical review letters 109.15 (2012): 150501.

\bibitem{vlachou2015quantum}
Vlachou, Chrysoula, et al. ``Quantum walk public-key cryptographic system." International Journal of Quantum Information 13.07 (2015): 1550050.

\bibitem{vlachou2018quantum}
Vlachou, Chrysoula, et al. ``Quantum key distribution with quantum walks." Quantum Information Processing 17.11 (2018): 1-37.

\bibitem{srikara2020quantum}
Srikara, S., and C. M. Chandrashekar. ``Quantum direct communication protocols using discrete-time quantum walk." Quantum Information Processing 19.9 (2020): 1-15.

\bibitem{QW-QRNG}
Sarkar, Anupam, and C. M. Chandrashekar. ``Multi-bit quantum random number generation from a single qubit quantum walk." Scientific reports 9.1 (2019): 1-11.

\bibitem{bae2021quantum} 
Bae, Minwoo, and Walter O. Krawec. ``Semi-source independent quantum walk random number generation." 2021 IEEE Information Theory Workshop (ITW). IEEE, 2021.


%%%%%%%%%%%%%%%%%%%%%%%%%%%%%%%%%%%%%%%%%%%%%%%%%%%%%%%%%%%%%%%%%%%%
% security analysis: uncertainty principle and quantum sampling
\bibitem{heisenberg1927uncertainty}
W. Heisenberg, ``Uber den anschaulichen inhalt der quan-tentheoretischen kinematik und mechanik," Z. Physik 43, 172 (1927).

\bibitem{robertson1929uncertainty}
Robertson, Howard Percy. ``The uncertainty principle." Physical Review 34.1 (1929): 163.

\bibitem{hirschman1957uncertainty}
I. I. Hirschman. ``A Note on Entropy." Am. J. Math., 79(1): 152–156, 1957.

\bibitem{deutsh1983uncertainty}
D. Deutsch. ``Uncertainty in Quantum Measurements." Physical review letters  50(9): 631–633, 1983.

\bibitem{maassen1988uncertainty}
Maassen, Hans, and Jos BM Uffink. ``Generalized entropic uncertainty relations." Physical review letters  60.12 (1988): 1103.

\bibitem{krishna2002uncertainty}
Krishna, M., and K. R. Parthasarathy. ``An entropic uncertainty principle for quantum measurements." Sankhyā: The Indian Journal of Statistics, Series A (2002): 842-851.

\bibitem{renner2008security}
Renner, Renato. ``Security of quantum key distribution." International Journal of Quantum Information 6.01 (2008): 1-127.

\bibitem{renes2011security}  
Tomamichel, Marco, et al. ``Leftover hashing against quantum side information." IEEE Transactions on Information Theory 57.8 (2011): 5524-5535.

\bibitem{renes2012security}  
Renes, Joseph M., and Renato Renner. ``One-shot classical data compression with quantum side information and the distillation of common randomness or secret keys." IEEE Transactions on Information Theory 58.3 (2012): 1985-1991.  
  
\bibitem{berta2010uncertainty}
Berta, Mario, et al. ``The uncertainty principle in the presence of quantum memory." Nature Physics 6.9 (2010): 659-662.

\bibitem{coles2011uncertainty}
Coles, Patrick J., Li Yu, and Michael Zwolak. ``Relative entropy derivation of the uncertainty principle with quantum side information." arXiv preprint arXiv:1105.4865 (2011).

\bibitem{tomamichel2011uncertainty}
Tomamichel, Marco, and Renato Renner. ``Uncertainty relation for smooth entropies." Physical review letters 106.11 (2011): 110506.

\bibitem{tomamichel2012thesis}
Tomamichel, Marco. ``A framework for non-asymptotic quantum information theory." arXiv preprint arXiv:1203.2142 (2012).

\bibitem{tomamichel2013hanggi}
Tomamichel, Marco, and Esther Hänggi. ``The link between entropic uncertainty and nonlocality." Journal of Physics A: Mathematical and Theoretical 46.5 (2013): 055301.

\bibitem{coles2012uncertainty}
Coles, Patrick J., et al. ``Uncertainty relations from simple entropic properties." Physical review letters 108.21 (2012): 210405.

\bibitem{coles2014uncertainty}
Coles, Patrick J., and Marco Piani. ``Improved entropic uncertainty relations and information exclusion relations." Physical Review A 89.2 (2014): 022112.

\bibitem{uncertainty2017survey} %ent-survey
Coles, Patrick J., et al. ``Entropic uncertainty relations and their applications." Reviews of Modern Physics 89.1 (2017): 015002.
    
\bibitem{krawec2019quantum}
Krawec, Walter O. ``Quantum sampling and entropic uncertainty." Quantum Information Processing 18.12 (2019): 1-18.
  
\bibitem{krawec2020new}
Krawec, Walter O. ``A new high-dimensional quantum entropic uncertainty relation with applications." 2020 IEEE International Symposium on Information Theory (ISIT). IEEE, 2020.

%%%%%%%%%%%%%%%%%%%%%%%%%%%%%%%%%%%%%%%%%%%%%%%%%%%%%%%%%%%%%%%%%%%%
% quantum sampling
  
\bibitem{quantum_sampling}
Bouman, Niek J., and Serge Fehr. ``Sampling in a quantum population, and applications." Annual Cryptology Conference. Springer, Berlin, Heidelberg, 2010.
  
%%%%%%%%%%%%%%%%%%%%%%%%%%%%%%%%%%%%%%%%%%%%%%%%%%%%%%%%%%%%%%%%%%%%
% operator norms
\bibitem{roger2012matrix}
Horn, Roger A., and Charles R. Johnson. Matrix analysis. Cambridge university press, 2012.


%%%%%%%%%%%%%%%%%%%%%%%%%%%%%%%%%%%%%%%%%%%%%%%%%%%%%%%%%%%%%%%%%%%%
% nonlocal experiments/games
\bibitem{clauser1969chsh}
Clauser, John F., et al. ``Proposed experiment to test local hidden-variable theories." Physical review letters 23.15 (1969): 880.

\bibitem{francesco2012nonlocal}
Buscemi, Francesco. ``All entangled quantum states are nonlocal." Physical review letters 108.20 (2012): 200401.

\bibitem{matthew2017nonlocal}
McKague, Matthew. ``Self-testing in parallel with CHSH." Quantum 1 (2017): 1.

\bibitem{ming2019nonlocal}
Luo, Ming-Xing. ``A nonlocal game for witnessing quantum networks." npj Quantum Information 5.1 (2019): 1-6.

\bibitem{johnny2022nonlocal}
Hooyberghs, Johnny. ``The CHSH Game." Introducing Microsoft Quantum Computing for Developers. Apress, Berkeley, CA, 2022. 271-303.


%%%%%%%%%%%%%%%%%%%%%%%%%%%%%%%%%%%%%%%%%%%%%%%%%%%%%%%%%%%%%%%%%%%%
% freedom of choice/measurement indepent

\bibitem{thomas2010freedom_choice}
Scheidl, Thomas, et al. ``Violation of local realism with freedom of choice." Proceedings of the National Academy of Sciences 107.46 (2010): 19708-19713.

\bibitem{bernhard2012freedom_choice}
Wittmann, Bernhard, et al. ``Loophole-free Einstein–Podolsky–Rosen experiment via quantum steering." New Journal of Physics 14.5 (2012): 053030.

\bibitem{andrew2019freedom_choice}
Friedman, Andrew S., et al. ``Relaxed Bell inequalities with arbitrary measurement dependence for each observer." Physical Review A 99.1 (2019): 012121.

\bibitem{michael2020freedom_choice}
Hall, Michael JW, and Cyril Branciard. ``Measurement-dependence cost for Bell nonlocality: Causal versus retrocausal models." Physical Review A 102.5 (2020): 052228.

\bibitem{rafael2021freedom_choice}
Chaves, Rafael, et al. ``Causal networks and freedom of choice in bell’s theorem." PRX Quantum 2.4 (2021): 040323.


%%%%%%%%%%%%%%%%%%%%%%%%%%%%%%%%%%%%%%%%%%%%%%%%%%%%%%%%%%%%%%%%%%%%
% device-independent qrng/qkd
\bibitem{artur1992qc_bell}
Ekert, Artur K. ``Quantum Cryptography and Bell’s Theorem." Quantum Measurements in Optics. Springer, Boston, MA, 1992. 413-418.

\bibitem{stefano2009diqkd}
Pironio, Stefano, et al. ``Device-independent quantum key distribution secure against collective attacks." New Journal of Physics 11.4 (2009): 045021.

\bibitem{lluis2011diqkd}
Masanes, Lluis, Stefano Pironio, and Antonio Acín. ``Secure device-independent quantum key distribution with causally independent measurement devices." Nature communications 2.1 (2011): 1-7.

\bibitem{charles2013diqkd}
Lim, Charles Ci Wen, et al. ``Device-independent quantum key distribution with local Bell test." Physical Review X 3.3 (2013): 031006.

\bibitem{yang2018diqrng}
Liu, Yang, et al. ``Device-independent quantum random-number generation." Nature 562.7728 (2018): 548-551.

\bibitem{umesh2019diqkd}
Vazirani, Umesh, and Thomas Vidick. ``Fully device independent quantum key distribution." Communications of the ACM 62.4 (2019): 133-133.

\bibitem{yanbao2020diqrng}
Zhang, Yanbao, et al. ``Experimental low-latency device-independent quantum randomness." Physical review letters 124.1 (2020): 010505.

\bibitem{david2021diqkd}
Nadlinger, David P., et al. ``Device-independent quantum key distribution." arXiv preprint arXiv:2109.14600 (2021).

 
%%%%%%%%%%%%%%%%%%%%%%%%%%%%%%%%%%%%%%%%%%%%%%%%%%%%%%%%%%%%%%%%%%%%
% odd-even parity
\bibitem{wernsdorfer1999parity}
Wernsdorfer, W., and R. Sessoli. ``Quantum phase interference and parity effects in magnetic molecular clusters." science 284.5411 (1999): 133-135.

\bibitem{schmid2000parity}
Schmid, J., et al. ``Absence of odd-even parity behavior for Kondo resonances in quantum dots." Physical Review Letters 84.25 (2000): 5824.

\bibitem{kanghun2000parity}
Ahn, Kang-Hun, and Peter Fulde. ``Parity effects in stacked nanoscopic quantum rings." Physical Review B 62.8 (2000): R4813.

\bibitem{orellana2003parity}
Orellana, P. A., et al. ``Transport through a quantum wire with a side quantum-dot array." Physical Review B 67.8 (2003): 085321.

\bibitem{eichler2007parity}
Eichler, A., et al. ``Even-odd effect in Andreev transport through a carbon nanotube quantum dot." Physical review letters 99.12 (2007): 126602.

\bibitem{jun2011parity}
Su, Jun, Feng-Shou Zhang, and Bao-An Bian. ``Odd-even effect in heavy-ion collisions at intermediate energies." Physical Review C 83.1 (2011): 014608.

\bibitem{yao2012parity}
Lu, Yao, and Gui Lu Long. ``Parity effect and phase transitions in quantum Szilard engines." Physical Review E 85.1 (2012): 011125.

\bibitem{beenakker2013parity}
Beenakker, C. W. J., et al. ``Fermion-parity anomaly of the critical supercurrent in the quantum spin-hall effect." Physical review letters 110.1 (2013): 017003.

\bibitem{thilagam2013parity}
Thilagam, A. ``Binding energies of composite boson clusters using the Szilard engine." arXiv preprint arXiv:1309.6493 (2013).

\bibitem{zekun2014parity}
Zhuang, Zekun, and Shi-Dong Liang. ``Quantum Szilard engines with arbitrary spin." Physical Review E 90.5 (2014): 052117.

\bibitem{pal2016parity}
Pal, P. S., and A. M. Jayannavar. ``Role of partition in work extraction from multi-particle Szilard Engine." arXiv preprint arXiv:1612.07007 (2016).

\bibitem{jun2019parity}
Yin, Jun, et al. ``Dimensional reduction, quantum Hall effect and layer parity in graphite films." Nature Physics 15.5 (2019): 437-442.

\bibitem{petr2019parity}
Stepanov, Petr, et al. ``Quantum parity Hall effect in Bernal-stacked trilayer graphene." Proceedings of the National Academy of Sciences 116.21 (2019): 10286-10290.

\bibitem{razmadze2020parity}
Razmadze, D., et al. ``Quantum dot parity effects in trivial and topological Josephson junctions." Physical Review Letters 125.11 (2020): 116803.

\bibitem{parijat2020parity}
Banerjee, Parijat, et al. ``Quantum quench and thermalization to GGE in arbitrary dimensions and the odd-even effect." Journal of High Energy Physics 2020.9 (2020): 1-51.

\bibitem{congyi2021parity}
Mou, Congyi, Shuhao Cao, and Tao Zhou. ``The optimal Maxwell's demon in quantum Szilard engine." Quantum Engineering 3.4 (2021): e82.

\bibitem{yuanye2022parity}
Yin, Yuanye, et al. ``Alkyl-Engineered Dual-State Luminogens with Pronounced Odd–Even Effects: Quantum Yields with up to 48\% Difference and Crystallochromy with up to 22 nm Difference." The Journal of Physical Chemistry B 126.15 (2022): 2921-2929.

\bibitem{arindam2022parity}
Sadhu, Arindam, et al. ``Low power design methodology in quantum dot cellular automata." Computers and Electrical Engineering 97 (2022): 107638.
%
\bibitem{william2022parity}
Livingston, William P., et al. ``Experimental demonstration of continuous quantum error correction." Nature communications 13.1 (2022): 1-7.

\bibitem{shaikh2022parity}
Sacco Shaikh, Daniel, Maura Sassetti, and Niccolò Traverso Ziani. ``Parity-Dependent Quantum Phase Transition in the Quantum Ising Chain in a Transverse Field." Symmetry 14.5 (2022): 996.

\bibitem{sergeevich2022parity}
Podoshvedov, Mikhail Sergeevich, et al. ``Promising Quantum Engineering of Optical Even/Odd Schrodinger Cat States." Bulletin of the South Ural State University. Series: Mathematics. Mechanics. Physics. 14.1 (2022): 77-85.

\bibitem{sassetti2022parity}
Sassetti, M. ``Parity Dependent Quantum Phase Transition in the Quantum Ising Chain in a Transverse Field." (2022).


%==================================
%\bibitem{quantum_walk_graphs} 
%Aharonov, Dorit, et al. ``Quantum walks on graphs." Proceedings of the thirty-third annual ACM symposium on Theory of computing. 2001.
%
%\bibitem{quantum_walk_cycles} 
%Bednarska, Małgorzata, et al. ``Quantum walks on cycles." Physics Letters A 317.1-2 (2003): 21-25.

%\bibitem{limit_thm_qwm} 
%Konno, Norio, and Takuya Machida. ``Limit theorems for quantum walks with memory." arXiv preprint arXiv:1004.0443 (2010).

%\bibitem{intrinsic_randomness} 
%Yuan, Xiao, et al. ``Intrinsic randomness as a measure of quantum coherence." Physical Review A 92.2 (2015): 022124.

%\bibitem{op_meaning}
%Konig, Robert, Renato Renner, and Christian Schaffner. ``The operational meaning of min-and max-entropy." IEEE Transactions on Information theory 55.9 (2009): 4337-4347.


\end{thebibliography}
\end{document}